\documentclass[letterpaper,11pt,english]{article}

\usepackage[letterpaper, portrait, margin=1in]{geometry}

\usepackage{amsmath,amsthm}\allowdisplaybreaks
\usepackage{amssymb}
\usepackage{mathtools}
\usepackage{csquotes}
\usepackage{graphicx}
\usepackage{grffile}
\usepackage{subfig}
\usepackage[hidelinks]{hyperref}

\usepackage{soul}
\usepackage{xcolor}
\usepackage{xspace}
\usepackage{paralist}
 
\usepackage{algorithm}
\usepackage[noend]{algorithmic}
\usepackage{textcomp}
\usepackage{nicefrac}

 \usepackage[disable]{todonotes} 
\usepackage[subtle]{savetrees}

\theoremstyle{plain}
\newtheorem{theorem}{Theorem}
\newtheorem{lemma}[theorem]{Lemma}

\newtheorem{corollary}[theorem]{Corollary}

\newcommand*{\OPT}{\ensuremath{\textsc{Opt}}\xspace}
\newcommand*{\ALG}{\ensuremath{\textsc{Alg}}\xspace}

\newcommand{\workloadSmall}{W^{\text{min}}}

\newcommand{\itemB}{B}
\newcommand{\itemL}{L}
\newcommand{\itemS}{S}
\newcommand{\itemO}{O}

\newcommand{\binType}{\ensuremath{T}}
\newcommand{\binVar}{\binType}
\newcommand{\binSet}{\ensuremath{\mathcal{T}}}

\newcommand{\binBL}{BL}
\newcommand{\binBS}{BS}
\newcommand{\binB}{B}
\newcommand{\binLLS}{LLS}
\newcommand{\binLL}{LL}
\newcommand{\binLSS}{LSS}
\newcommand{\binLS}{LS}
\newcommand{\binL}{L}
\newcommand{\binSSS}{SSS}
\newcommand{\binSS}{SS}
\newcommand{\binS}{S}

\newcommand{\quartet}{Q^\ALG}
\newcommand{\quartetLL}{\ALG_{\binLL}^Q}
\newcommand{\quartetSSS}{\ALG_{\binSSS}^Q}
\newcommand{\quartetLLSSS}{\ALG_{\binLL,\binSSS}^Q}
\newcommand{\nonquartetLL}{\ALG_{\binLL}^{-Q}}
\newcommand{\nonquartetSSS}{\ALG_{\binSSS}^{-Q}}

\newcommand{\NBO}[1]{\ensuremath{\OPT_{#1}}}
\newcommand{\NBA}[1]{\ensuremath{\ALG_{#1}}}

\newcommand{\aboveAndBelowSkip}{4pt}

\newenvironment{tightalign}
 {
 \setlength{\abovedisplayshortskip}{\aboveAndBelowSkip}\setlength{\abovedisplayskip}{\aboveAndBelowSkip}
 \setlength{\belowdisplayshortskip}{\aboveAndBelowSkip}\setlength{\belowdisplayskip}{\aboveAndBelowSkip}
 \align}
 {\endalign}
\newenvironment{tightalign*}
 {
 \setlength{\abovedisplayshortskip}{\aboveAndBelowSkip}\setlength{\abovedisplayskip}{\aboveAndBelowSkip}
 \setlength{\belowdisplayshortskip}{\aboveAndBelowSkip}\setlength{\belowdisplayskip}{\aboveAndBelowSkip}
 \csname align*\endcsname}
 {\endalign}

\usepackage{authblk}

\title{A Tight Approximation for Fully Dynamic Bin Packing\newline without Bundling
\thanks{
An extended abstract of this work, merged with~\cite{DBLP:journals/corr/abs-1711-02078}, will appear in ICALP 2018.
\newline
This work is partially supported by the German Research Foundation (DFG) within the Collaborative Research Center ``On-The-Fly Computing'' (SFB 901).}}
\author[]{Bj\"orn Feldkord}
\author[]{Matthias Feldotto}
\author[]{S\"oren Riechers}
\affil[]{Heinz Nixdorf Institute and Department of Computer Science\\
	Paderborn University, F\"urstenallee 11, 33102 Paderborn, Germany
}
\affil[]{ \{bjoernf,feldi,soerenri\}@mail.upb.de}
\date{}

\begin{document}
\hypersetup{pageanchor=false}
\maketitle

\begin{abstract}
We consider a variant of the classical Bin Packing Problem, called Fully Dynamic Bin Packing.
In this variant, items of a size in $(0,1]$ must be packed in bins of unit size. In each time step, an item either arrives or departs from the packing. An algorithm for this problem must maintain a feasible packing while only repacking a bounded number of items in each time step.

We develop an algorithm which repacks only a constant number of items per time step and, unlike previous work, does not rely on bundling of small items which allowed those solutions to move an unbounded number of small items as one. Our algorithm has an asymptotic approximation ratio of roughly $1.3871$ which is complemented by a lower bound of Balogh et al.~\cite{DBLP:journals/siamcomp/BaloghBGR08}, resulting in a tight approximation ratio for this problem.
As a direct corollary, we also close the gap to the lower bound of the Relaxed Online Bin Packing Problem in which only insertions of items occur.
\end{abstract}

\thispagestyle{empty}
\clearpage
\setcounter{page}{1}
\pagebreak
\hypersetup{pageanchor=true}


\section{Introduction}\label{sec:Introduction}

A problem instance of the classical \emph{Bin Packing Problem} is given by a list of $n$ items $L=\left(a_1, a_2, \ldots, a_n\right)$.
The objective is to pack these items into as few bins (of unit size) as possible.
The problem is known to be NP-complete~\cite{DBLP:books/fm/GareyJ79} and one of the most studied problems in theoretical computer science since it relates to many applications in areas such as load balancing and file management.
For the offline version both an APTAS~\cite{DBLP:journals/combinatorica/VegaL81} and an AFPTAS~\cite{DBLP:conf/focs/KarmarkarK82} have been developed a long time ago.
Recent research has concentrated on the online version and other more dynamic behavior, thus covering a wider range of scenarios with this model.

In many realistic settings, items arrive over time instead of having the whole input available at the beginning of the computation.
While online problems, specifically the \emph{Online Bin Packing Problem}~\cite{princeton1971performance}, cover this property the reality often goes beyond.
In addition to the online arrival of items, some items may also depart from the packing over time.
As an example, when minimizing the number of machines for a collection of long-time running jobs, some new jobs may be added while others may be terminated.
Also, the number of machines might be reduced by migrating some jobs from one machine to another.
However, since migrating jobs causes additional costs, the number of moved jobs should be bounded.
Motivated by such examples and unlike the classical setting, in bin packing models with repacking, the assignment of items to bins is not irrevocable.
Instead, we allow an algorithm to repack a bounded number of items in each time step.
The \emph{Fully Dynamic Bin Packing Problem}~\cite{DBLP:journals/siamcomp/IvkovicL98} represents this setting by requiring an algorithm
to react to dynamic changes in the input while repacking only a small amount of items in every time step.

Due to the integer nature of bin packing problems and because we are typically interested in large instances,
the \emph{asymptotic approximation ratio} has been used to measure the quality of an algorithm in the Fully Dynamic Bin Packing Problem:
Let $L_t$ be the input sequence at time step $t$, $\ALG(L_t)$ the number of bins used by the algorithm and $\OPT(L_t)$ the number of bins used by an optimum solution.
The asymptotic approximation ratio is defined as
 $\lim_{x\rightarrow\infty}\sup_{\OPT(L_t)=x}\frac{\ALG(L_t)}{\OPT(L_t)}$.

Regarding the repacking after each modification of the input,
it is practically infeasible to allow an algorithm to move an arbitrary amount of items.
Two different approaches have been developed to measure the amount of repacking:
One line of research focuses on the migration factor~\cite{DBLP:journals/mor/SandersSS09}, which is defined as the ratio between the total load of moved items during an insertion (deletion) and the size of the new (removed) element.
Our work joins the second line of research where the number of shifting moves is bounded~\cite{DBLP:journals/siamcomp/GambosiPT00}: that is, only a constant (absolute) number of items may be moved from one bin to another after each insertion (or deletion).

One major challenge of Fully Dynamic Bin Packing in the algorithmic context is the presence of very small and very large items.
Deleting (or inserting) one large item in a packing with many small items may result in bins becoming quite empty while still holding many items.
Consequently, this requires lots of repacking (see also the lower bounds with constructions based on this fact~\cite{DBLP:journals/siamcomp/BaloghBGR08,DBLP:journals/ipl/IvkovicL96}).
To overcome this issue,
the technique of \emph{bundling} was introduced~\cite{DBLP:journals/siamcomp/GambosiPT00} and is widely used in the literature.
Instead of handling each item separately, small items with a size below some threshold are grouped together and handled like one element.
This notion is reasonable in the context of the migration factor as well as in situations where the handled items may be regarded as similar such that
multiple items can essentially be treated as one.
However, if items are more unique and the processing of one shift within the respective application, for example setting up a task on a virtual machine,
requires effort (mostly) independent of the item's size, then a more strict measurement of repacking is needed to accurately assess the efficiency of the algorithm.
The problem of not utilizing bundling is also interesting from a theoretical point of view since a tight ratio for this problem together with the existing $(1+\varepsilon)$-approximation for the problem with bundling~\cite{DBLP:conf/approx/BerndtJK15} does
allow us to precisely judge the improvement achievable through the bundling technique.


\subsection{Our Contribution}\label{sec:Contribution}
We provide an algorithm for the Fully Dynamic Bin Packing Problem which repacks only a constant number of items after each insertion or deletion.
To the best of our knowledge, this is the first algorithm with a constant approximation ratio, using a constant, non-amortized number of repackings, that does not use a bundling technique for handling very small items.
Furthermore, we achieve tightness on the asymptotic approximation ratio,
for which our algorithm approaches $\alpha:=1-\nicefrac{1}{\left(W_{-1}\left(\nicefrac{-2}{e^3}\right)+1\right)}\approx 1.3871$ while there exists a matching
lower bound for this problem~\cite{DBLP:journals/computing/BaloghBGM09,DBLP:journals/siamcomp/BaloghBGR08}.
Here, $W_{-1}$ denotes the lower branch of the Lambert-W-function.
Apart from the fully dynamic version with bundling, this is the first variant of dynamic or online bin packing in which matching lower and upper bounds could be shown.
Additionally, we improve the best known algorithm for Relaxed Online Bin Packing~\cite{DBLP:journals/jco/BaloghBGR14} and also close the gap between lower and upper bound in this model.

We tackle the problem of restricted repacking in three steps.
First, we pack all items with a size below a very small threshold, but bins are filled only to a varying height, thus reserving spaces of different size for potential large items.
The size of these spaces are carefully chosen to counter bad instances where large items do not fit into them.
Next, we pack the items which do not fall in the first category such that the packing without the small items fulfills certain structural properties which are helpful for the analysis.
In the third step, we carefully merge the bins such that the number of moved items per time step remains constant while still guaranteeing the desired quality of the overall packing.

Our main result can be summarized in the following theorem:
\begin{theorem}\label{theorem:main_result}
For each $\varepsilon \in (0,1)$, there exists an algorithm for the Fully Dynamic Bin Packing Problem with
an asymptotic approximation ratio of $\left(1+\varepsilon\right)\cdot\left(1-\nicefrac{1}{\left(W_{-1}\left(\nicefrac{-2}{e^3}\right)+1\right)}\right)\approx \left(1+\varepsilon\right)\cdot 1.3871$
which repacks at most $\mathcal{O}\left(\nicefrac{1}{\varepsilon^2}\right)$ items per insertion or deletion of an item.
\end{theorem}


\subsection{Related Work}\label{sec:RelatedWork}

Based on the classical Bin Packing Problem~\cite{DBLP:books/fm/GareyJ79}, different online and dynamic variants have been developed and investigated.
Due to space constraints we focus on the scenarios which share the most important properties with our model and only mention the best results for them.
In the \emph{Online Bin Packing Problem}~\cite{princeton1971performance}, the items are unknown to the algorithm at the beginning and appear one after another.
Balogh et al.~\cite{DBLP:journals/tcs/BaloghBG12} show a lower bound of $1.54037$ for this problem and Seiden~\cite{DBLP:journals/jacm/Seiden02} presents an approximation algorithm with a competitive ratio of $1.58889$.
In the \emph{Dynamic Bin Packing} setting~\cite{DBLP:journals/siamcomp/CoffmanGJ83}, additionally to arrivals as in the Online Bin Packing Problem departures of items are also allowed.
Here, a lower bound of $\nicefrac{8}{3}$ by Wong et al.~\cite{DBLP:conf/isaac/WongYB12} and an upper bound of $2.897$ by Coffman et al.~\cite{DBLP:journals/siamcomp/CoffmanGJ83} exist.

We now turn our attention to the models which allow repacking of items.
Our main focus is the counting of shifting moves; for the alternative research line with the migration factor as the main cost function we refer to~\cite{DBLP:conf/approx/BerndtJK15}.
In the \emph{Relaxed Online Bin Packing Problem}~\cite{DBLP:journals/siamcomp/GambosiPT00}, online arrivals and no departures occur, but repacking of items is allowed.
Repacking means that items can be assigned to another bin in the course of the execution, while in a setting without repacking decisions are irrevocable.
The best known lower bound for an algorithm which uses only a constant number of shifting moves is originally given for our model, but it also applies to this setting with $1.3871$ by Balogh et al.~\cite{DBLP:journals/siamcomp/BaloghBGR08}.
From the positive perspective Balogh et al.~\cite{DBLP:journals/jco/BaloghBGR14} give an approximation algorithm based on the
Harmonic Fit Algorithm~\cite{DBLP:journals/jacm/LeeL85} for which they achieve a competitive ratio of $\nicefrac{3}{2}$.
Since our algorithm is also applicable to this setting we improve this result to also close the gap between the lower and upper bound for this problem.

Our setting is the most powerful model among the presented ones, the \emph{Fully Dynamic Bin Packing}~\cite{DBLP:journals/siamcomp/IvkovicL98}, in which we allow arrivals and departures of items as well as repacking.
Ivkovic and Lloyd~\cite{DBLP:journals/siamcomp/IvkovicL98,ivkovic2009fully} introduced the model of Fully Dynamic Bin Packing and developed an algorithm called \emph{Mostly Myopic Packing (MMP)} which achieves a $\nicefrac{5}{4}$-competitive ratio.
Their algorithm is based on an offline algorithm by Johnson~\cite{Johnson1973,DBLP:journals/jcss/Johnson74} and utilizes a technique whereby the packing of an item is done with a total disregard for already packed items of a smaller size.
In contrast to our work, they use the concept of bundling very small elements in their analysis.
They can show that the number of single items or bundles of very small elements that need to be repacked is bounded by a constant.
Additionally, Ivkovic~\cite{Ivkovic1996} also gives a slightly simpler version of this algorithm, called \emph{Myopic Packing (MP)}.
It uses similar ideas but ignores one step of MMP that results in a much easier analysis and a competitive ratio of $\nicefrac{4}{3}$.
Berndt et al.~\cite{DBLP:conf/approx/BerndtJK15} consider exactly the same setting, also allowing the bundling of very small elements in the analysis.
Their algorithm achieves a $1+\varepsilon$ approximation ratio and a bound of $\mathcal{O}\left(\nicefrac{1}{\varepsilon^4} \log \left(\nicefrac{1}{\varepsilon}\right)\right)$ 
for both the migration factor and the number of shifting moves.
The closest result to our work is a paper by Gupta et al.~\cite{DBLP:journals/corr/abs-1711-02078} who, independently of our work, give an algorithm which does not utilize bundling and matches the lower bound for the approximation ratio while only repacking $\mathcal{O}(\nicefrac{1}{\varepsilon^{4}})$ items per time step.

In addition to the positive algorithmic results, researchers also explored lower bounds for this setting.
All results assume that no bundling is allowed (otherwise there can only be the trivial lower bound of one~\cite{DBLP:conf/approx/BerndtJK15}), hence allowing only a constant number of shifting moves per time step.
Ivkovic and Lloyd~\cite{DBLP:journals/ipl/IvkovicL96} show a lower bound of $\nicefrac{4}{3}$.
Their construction uses the inability of an algorithm to react to insertions and deletions of items with size slightly larger than $\nicefrac{1}{2}$ when
the remaining items may be of arbitrarily small size.
Balogh~et~al.~\cite{DBLP:journals/computing/BaloghBGM09,DBLP:journals/siamcomp/BaloghBGR08} improve this bound to roughly $1.3871$.
They extend the technique of the previous lower bound by constructing multiple lists of large items whose sizes are chosen through the construction of a linear program.
Their results are the inspiration for some of the parameter choices in this work.
Gupta et al.~\cite{DBLP:journals/corr/abs-1711-02078} extend the bound by introducing a tradeoff between repacking and the additive term in the competitive ratio when
surpassing the lower bound for the asymptotic ratio.


\subsection{Model}

In the Fully Dynamic Bin Packing Problem, we are given a list $L_t=\left(a_{t,1},a_{t,2},\ldots,a_{t,n_t}\right)$ of items at time step $t$.
We write $a\in L_t$ if $a=a_{t,i}$ for some $i\in\{1,\ldots,n_t\}$ and denote by $size(a)\in(0,1]$ the size of item $a$.
For convenience, we abuse notation to replace $size(a)$ with $a$ wherever the meaning is clear from the context.
The initial list $L_0=()$ is the empty list and between two time steps $t$ and $t+1$, the lists $L_t$ and $L_{t+1}$
differ by at most one item, i.e. either at most one item arrives or departs in step $t+1$.

An algorithm for the problem must output a valid packing $\left(B_{t,1},B_{t,2},\ldots,B_{t, m_t}\right)$ for each step $t$ where
$\sum_{a\in B_{t,i}}size(a)\leq 1$ for all $i\in\{1,\ldots,m_t\}$ and for each $a\in L_t$ there is exactly one $i\in\{1,\ldots,m_t\}$ such that
$a\in B_{i,t}$.
In the following, we omit the index $t$ whenever it is clear from the context.
The number of used bins in step $t$ is the number of bins $B_i$ for which $size(B_i):=\sum_{a\in B_i}size(a)>0$.
The number of shifting moves between two steps $t$ and $t+1$ is the number of items $a$ for which the following properties hold:
$a\in L_t$, $a\in L_{t+1}$ and $a$ is placed in two different bins in steps $t$ and $t+1$. 


\section{Algorithmic Approach}
\label{sec:approach}

Our algorithm can be split into three main parts, reflected by the three upcoming sections.
The first two sections pack only minor and major items\footnote{The terms \emph{minor} and \emph{major} have no meaning of priority in our work, they only serve as a clear distinction between small items (size of at most $\nicefrac{\varepsilon}{15}$) and large items.} separately, while the third part
of the algorithm combines these two solutions by merging appropriate bins into one.

We classify items as minor if they have a size of at most $\nicefrac{\varepsilon}{15}$.
These items are handled almost independently since the cumulative size of these items
we are allowed to shift in every time step may be arbitrarily small.
The bins that our algorithm will fill with minor items have a certain objective height: that is, some of the bins are filled with minor items only to a certain threshold of at most $\nicefrac{1}{2}$, reserving the remaining space for major items.
In order to achieve an approximation ratio that is arbitrarily close to the lower bound, it is not sufficient to fill all these half-empty bins up to the same height (this only results in an approximation ratio of $\sqrt{2}$).
Instead, we must maintain a number of bin types, with each bin type having a different filling height, and at the same time we must ensure that for each bin type the fraction of bins of this type remains roughly the same over time.

In order to get these filling heights and percentages right, we use similar values to the parameters derived from the linear program in the lower bound of Balogh~et~al.~\cite{DBLP:journals/siamcomp/BaloghBGR08}.
For some small $\varepsilon$, the lower bound cleverly chooses $k$ different item sizes $s_1+\varepsilon,\ldots,s_k+\varepsilon$ in $(\nicefrac{1}{2},\nicefrac{3}{4})$ to force any algorithm to open a new bin if too many items of size $s_i+\varepsilon$ arrive.
For our upper bound, we basically use these item sizes to derive bin types.
Essentially, the remaining space in bin type $i$ is chosen to be $s_{i+1}$: that is, we maintain enough bins with a remaining space of $s_{i+1}$ to host the required number of items of size $s_i+\varepsilon$, but at the same time prevent our algorithm from failing if the adversary adds items of a size in $(s_i+\varepsilon, s_{i+1}]$ instead. 
In order to get roughly the correct shares for each bin type (we will later show that rounding these shares in a certain way does not hurt our solution too much), we introduce so-called \emph{bin groups}.
These bin groups of size $l=\mathcal{O}(\nicefrac{1}{\varepsilon})$ are always composed in the same way with regard to the bin types.

Inspired by ideas from~\cite{DBLP:conf/approx/BerndtJK15}, our algorithm maintains an order of the minor items within the set of bins: that is, for any two adjacent bins, all items in the left bin are at least as large as the items in the right bin.
Whenever a new item arrives, this item is added in a bin such that this order of items is maintained.
If this bin becomes overfull (in the way that the objective height is exceeded), the largest item in this bin is moved to the next bin to the left (which only hosts items that are at least as large as the moved item).
This shifting is repeated until no bin is overfull anymore.
However, to avoid this process iterating over up to $\mathcal{O}(\varepsilon n)$ bins, we introduce so-called \emph{buffer groups}: that is, a set of bins meant to serve as a buffer in the way that at least one bin in this group can store additional items without exceeding the objective height. 
Our algorithm ensures that the number of bin groups between two buffer groups is at most $\mathcal{O}(\nicefrac{1}{\varepsilon})$, implying a maximum number of shifted items of $\mathcal{O}(\nicefrac{1}{\varepsilon^2})$ per insertion.
The deletion of items is handled in a similar way, where the bin from which an item was removed draws the smallest item from its neighbor into it.

For the major items, we utilize an algorithm called Myopic Packing (MP) by Ivkovic~\cite{Ivkovic1996}.
This algorithm has a competitive ratio of $\nicefrac{4}{3}$ (it is below the lower bound for our model since it uses bundling)
and modifies only a constant number of bins per time step.
Applying this algorithm to only items of a size of at least $\nicefrac{\varepsilon}{15}$ restricts the amount of repacking
to $\mathcal{O}(\nicefrac{1}{\varepsilon})$ per time step.
We develop a new view on this algorithm to derive structural properties of the solution which are needed for the combination of major and minor items.
Note that there are other algorithms that we could have used for this part instead, in particular the improved version of the MP algorithm MMP~\cite{DBLP:journals/siamcomp/IvkovicL98,ivkovic2009fully} or the $(1+\varepsilon)$-algorithm by Berndt~et~al.~\cite{DBLP:conf/approx/BerndtJK15}.
However, we refrained from doing so in order to keep the algorithms as simple as possible while still achieving the desired tight approximation ratio.

Finally, the bins with minor items are merged with the bins with major items by utilizing a greedy-like approach, where small chunks of reserved space and a big cumulative size of major items is prioritized in order to guarantee an efficient utilization of the reserved space.
The combination has two main challenges: Firstly, we ensure that this greedy process only has to modify $\mathcal{O}(\nicefrac{1}{\varepsilon})$ bins per time step.
Secondly, we guarantee a space efficient combination resulting in an overall good solution quality.
The analysis carefully utilizes the structural insights about the solution for major items to estimate the solution quality within the bins that did not get merged.

\newpage
\section{Prospective Packing of Minor Items}
\label{sec:small}

In this section,
we provide an algorithm which handles only items of a size of at most $\delta:=\nicefrac{\varepsilon}{15}$.
Let $\workloadSmall=\sum_{a \in L_t, a \leq \nicefrac{\varepsilon}{15}} a$  be the workload of all minor items at time $t$.
Given by an input list of only minor items, we describe the current packing as an ordered list of $m$ bins $B_1,\ldots, B_l, B_{l+1}, \ldots, B_{2l}, \ldots,$ $B_{m-l+1}, \ldots, B_m$.
The bins are always handled in groups of $l$ neighboring bins.

To enable a later combination with major items (see Section~\ref{sec:Combination}), we will use only a sub-part of each bin for the minor items.
Therefore, each bin is assigned a type $j\in\{1,\ldots,k\}$, which specifies the desired filling height.
In each bin of type $j$, the load of the minor items will sum up to at most a given $w_j \leq 1$.
The concrete assignment of bins to these types as well as the specific values for all $w_j$ will be given later.

The algorithm keeps the minor items in the bins in a sorted order,
such that $a\geq a'$ for $a\in B_i,a'\in B_{i'}\ \forall i\leq i'$.
We consider a bin $B_i$ of type $j$ to be a \emph{full bin}
if $\sum_{e\in B_i}e+\min_{d \in B_{<i}}d>w_j$,
i.e. if no item from any bin preceding $B_i$ fits into $B_i$.
If a group of $l$ bins contains only bins that are full, we call it a \emph{full group}.
In the other case, if at least one bin is not full, it is called a \emph{buffer group}.
The algorithm aims to have at least $l$ and at most $2\cdot l$ full groups between two buffer groups, which is maintained by the insertion and deletion procedures.
The algorithm is initialized with an empty packing of $0$ items and a first buffer group with $l$ empty bins.

For notational convenience in the following description, we fix the index of the bin $B_i$ in which either an item is inserted into or deleted from while the other indizes are dynamically adapted i.e. a newly inserted bin to the left of $B_i$ immediately receives the index $i-1$.

\noindent\textbf{Insertion:}
A new item $a$ is added to the packing in the bin $B_i$ such that $e\geq a$ for all $e\in B_{i'}$, $i'<i$ and $e'\leq a$ for all $e'\in B_{i''}$, $i''>i$.
Let $B_i$ be a bin of type $j$.
We distinguish the following cases:
\begin{compactenum}
  \item $B_i$ is a full bin before the insertion. If $\sum_{e\in B_i}> w_j$ after the insertion, then we recursively insert $\max_{e\in B_i}e$ into $B_{i-1}$ and remove it from $B_i$.
	  Otherwise, the procedure terminates.
  \item $B_i$ is not full before and after the insertion of $a$. The procedure terminates.
	\item $B_i$ is not full before the insertion, but full after the insertion:
	  \begin{compactenum}
		\item
	  $B_i$ is not the left-most bin in a buffer group. If $\sum_{e\in B_i}e > w_i$ after the insertion, recursively insert $\max_{e\in B_i}e$ into $B_{i-1}$ and remove it from $B_i$, otherwise
		terminate.
		\item
		$B_i$ is the left-most bin in a buffer group. Insert a new buffer group to the left of $B_i$. If $\sum_{e\in B_i} e > w_j$ after the insertion, insert $\max_{e\in B_i}e$ into $B_{i-1}$ (i.e., the right-most bin of the new buffer group) and remove it from $B_i$.
		Additionally, if the distance between the group of $B_i$ and the next buffer group to the right is $2l$, insert a new buffer group to the right of $B_i$ such that there are $l$ groups between the inserted buffer group left to $B_i$ and the new buffer group.
		\end{compactenum}
\end{compactenum}

\begin{figure}[htb]
	\centering
	\includegraphics[scale=.34]{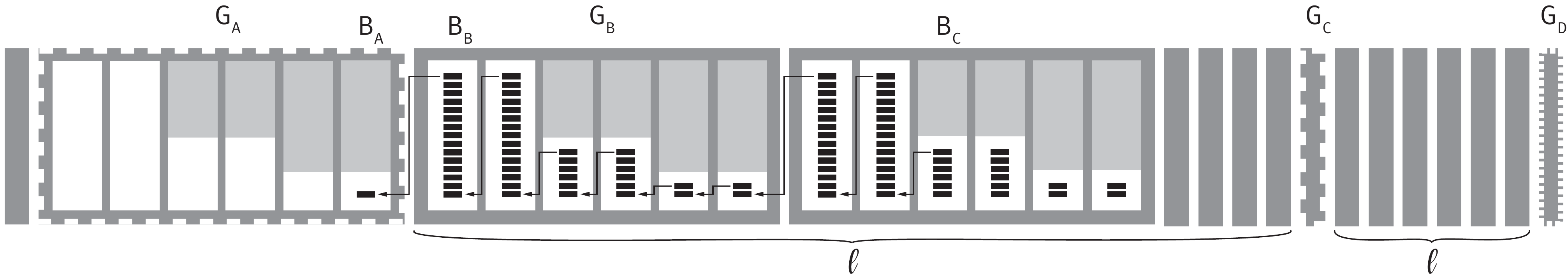}
	\caption{The displayed configuration shows the outcome of an insertion into bin $B_C$.
		Before this insertion, $G_B$ was a buffer group and the buffer groups $G_A$ and $G_C$ did not exist.
		An item was eventually shifted into $B_B$ after repeatedly running into case 1 in the recursion, causing the insertion of the new buffer group $G_A$ (cf. case 3b of the insertion procedure).
		$G_C$ also needs to be inserted since the distance between $G_A$ and $G_D$ now consists of $2l$ full groups.
		Afterwards, the recursive insertion of an item into bin $B_A$ causes the procedure to terminate (case 2).}
	\label{fig:algorithm_minor}
\end{figure}
\newpage
\noindent\textbf{Deletion:}
When an item $a$ is removed from a bin $B_i$, we proceed as follows:
\begin{compactenum}
  \item If $B_i$ is (still) full after the deletion, the procedure terminates.
	\item If $B_i$ was full before but is not full after the deletion and $B_{i-1}$ contains at least one item, then insert $\min_{e\in B_{i-1}}e$ into $B_i$ and recursively delete it from $B_{i-1}$.
	\item If $B_i$ is not full after the deletion and $B_{i-1}$ contains no item, then:
	  \begin{compactenum}
		  \item If $B_i$ is part of a buffer group, then terminate.
			\item If $B_i$ is not part of a buffer group, then remove the buffer group to the left of $B_i$ (since all bins in that group are empty now).
			  If there are at least $l$ full groups between the group of $B_i$ and the next buffer group to the right of it, then terminate (the group of $B_i$ is now a buffer group).
				Else, recursively delete $\min_{e\in B_{i-1}}e$ from $B_{i-1}$ (after the deletion of the former buffer group, $B_{i-1}$ now contains an item) and insert it into $B_i$.
				If there are now at least $2l$ full groups between the buffer group to the left and the buffer group to the right of the group of $B_i$, insert a new buffer group
        between those groups such that the distance of the new buffer group to those two groups is between $l$ and $2l$.
		\end{compactenum}
\end{compactenum}


\paragraph*{Choice of Parameters}\label{sec:parameters}
What remains open in the description of the algorithm is the concrete assignment of bin types and the choice of the parameters $k$ and $l$.
Our choice of filling heights is inspired by the parameters from the lower bound by Balogh et al.~\cite{DBLP:journals/siamcomp/BaloghBGR08},
but in order to get the desired upper bound instead, each filling height is essentially replaced by the next smaller filling height from the lower bound (see also the short discussion about this in Section~\ref{sec:approach}).
Let $\alpha:=1-\nicefrac{1}{\left(W_{-1}\left(\nicefrac{-2}{e^3}\right)+1\right)}\approx 1.3871$ be the value of the lower bound.

For each bin group of size $l$, we need to take care of the correct fraction of bins of a certain type $j$, which we implicitly determine by parameters $z_j$ (for notational convenience, we also write $z_j'\coloneqq \sum_{i=1}^{j}z_i$).
For a minor item workload of $\workloadSmall$, we aim to create roughly $z_k' \cdot \workloadSmall$ bins, where we choose $z_k' \coloneqq (1+\nicefrac{\varepsilon}{4})\alpha$, hence achieving the desired tight bound for minor items.
The filling height corresponding to bin type $k$ is defined as $w_k\coloneqq y_k\coloneqq \nicefrac{(z_k'-1)}{z_k'}$.
Bins of this type have the largest remaining space, which is $\nicefrac{1}{z_k'}$.
Intuitively, the reason is that we do not need to reserve space for major items of size at least $\nicefrac{1}{z_k'}$ as packing these items in exclusive bins still results in an approximation ratio of $z_k'$.

For the other bin types $j < k$, we now choose the parameters $y_j$ according to the geometric series $y_j = \frac{1}{2} \left(2\cdot y_k\right)^{\frac{j-1}{k-1}}$ (see also~\cite{DBLP:journals/siamcomp/BaloghBGR08} for the background on why this is a good choice).
The filling heights $w_j$ of the different bin types depend almost directly on these parameters: We set $w_j\coloneqq y_j\  \forall j > 1$ and $w_1\coloneqq 1 = y_1 + \frac{1}{2}$.
This perceived inconsistency is due to the shift of the other $w_j$ (w.r.t. the lower bound) as explained above.

The remaining values for $z_j$ are set such that $z_j'\coloneqq \nicefrac{y_k}{(y_j(1-y_k))}$ holds for all $j$. The values for $z_j'$ are a result of optimizing the number of bins of type $\leq j$ against a class of bad instances where many items of size $1-y_j+\varepsilon'$ (for some tiny $\varepsilon' > 0$) are inserted. These items can not be packed into the same bins with such types, however fit into bins of type $>j$.
Note that the choice of these parameters results in $\nicefrac{1}{4} < 1-\nicefrac{1}{z_k'} = y_k<\ldots<y_1=\nicefrac{1}{2}$ and $\nicefrac{3}{4} < 2\cdot (z_k'-1) = z_1'<\ldots<z_k'=(1+\nicefrac{\varepsilon}{4})\cdot\alpha$.

Based on these parameters, a group of $l$ bins is organized as follows:
We choose the size of a bin group to be $l=\left\lceil\nicefrac{4z_k'}{\varepsilon\alpha}\right\rceil=\lceil 4/\epsilon\rceil +1$
and the total number of bins of type $\leq j$ to be $\left\lceil\nicefrac{z_j'}{z_k'}\cdot l\right\rceil$.
Hence, the number of bins of type $1$ is $\left\lceil\nicefrac{z_{1}}{z_k'}\cdot l\right\rceil$, whereas for bin types $j>1$, it is determined by $\left\lceil\nicefrac{z_j'}{z_k'}\cdot l\right\rceil - \left\lceil\nicefrac{z_{j-1}'}{z_k'}\cdot l\right\rceil$.
Note that the rounding implies that the number of bins of some types may be zero.

Finally, we choose the number of bin types to be $k=\left\lceil\nicefrac{3}{\varepsilon}\right\rceil+1$.


\paragraph*{Analysis}
The main goal of the upcoming analysis is to bound the number of bins used for a given payload $\workloadSmall$ of minor items.
In the following we state properties of the structure of the algorithm's solution.
Their proofs and also further proofs can be found in the appendix.
\begin{lemma}\label{lemma:ordering_feasibility}
	During the whole execution of the algorithm it holds for two arbitrary bins $B_i$ and $B_{i'}$ with $i<i'$ that $\min_{e \in B_i}e \geq \max_{e \in B_{i'}}e$.
	Furthermore, for each bin $B$ with type $j$ it holds that $\sum_{e\in B}e\leq w_j$.
\end{lemma}

\begin{lemma}\label{lemma:distance}
	 The number of full groups between two buffer groups is in $[l,2l]$.
\end{lemma}

In order to count the number of bins in our solution, we need the following technical lemma.
Remember that $\delta = \nicefrac{\varepsilon}{15}$ denotes the maximum size of a minor item.

\begin{lemma}
\label{lemma:fitting}
A set of full bins which consists of at least $z_j\cdot\workloadSmall$ bins of type $j$ for each bin type $j$ contains a workload of at least $\workloadSmall$ of minor items,
i.e. $\sum_{j=1}^k z_j\workloadSmall \cdot (w_j - \delta) \geq \workloadSmall$.
\end{lemma}

\begin{lemma}
\label{lemma:smallmain}
For all $1\leq j\leq k$, the total number of bins of any type $i\leq j$ is at most $(z_j'+\frac{3}{4}\varepsilon\alpha)\workloadSmall$.
\end{lemma}

This bound is mainly used for the analysis in Section~\ref{sec:Combination}, however it also directly implies the approximation ratio
for instances where only minor items are present.

\begin{corollary}\label{lemma:small_approx}
	For a packing with only minor items, the algorithm achieves an approximation ratio of $z_k'+\frac{3}{4}\varepsilon\alpha=(1+\varepsilon)\alpha$.
\end{corollary}

\begin{lemma}\label{lemma:small_shifting_moves}
	The number of shifting moves (regarding minor items) during an insertion or deletion of a minor item is bounded by $\mathcal{O}(\nicefrac{1}{\varepsilon^2})$.
\end{lemma}


\section{Dealing with Major Items}
\label{sec:large}

For major items,
i.e. items with a size larger than $\delta=\nicefrac{\varepsilon}{15}$, we use an algorithm called \emph{Myopic Packing} (MP) by Ivkovic~\cite{Ivkovic1996} which is a simplified version of the \emph{MMP} algorithm by Ivkovic and Lloyd~\cite{DBLP:journals/siamcomp/IvkovicL98}.
The algorithm is essentially a fully dynamic variant of Johnson's First Fit Grouping Algorithm~\cite{DBLP:conf/focs/Johnson72,Johnson1973}.

We divide the major items in four sub-groups, depending on their size: A $\itemB$ (big) item has a size in $\left(\frac{1}{2}, 1\right]$, an $\itemL$ (large) item in  $\left(\frac{1}{3}, \frac{1}{2}\right]$, an $\itemS$ (small) item in $\left(\frac{1}{4}, \frac{1}{3}\right]$ and an $\itemO$ (other) item in  $\left(\frac{\varepsilon}{15}, \frac{1}{4}\right]$.
Ignoring additional $\itemO$ items for now, the following bin types of interest can occur:
$\binBL$, $\binBS$, $\binB$, $\binLLS$, $\binLL$, $\binLSS$ and $\binSSS$.
The name of the bin type represents the items of type $\itemB$, $\itemL$ and $\itemS$ contained in that bin.
The additional bin types $\binLS, \binL, \binSS, \binS$ can only occur at most two times in a packing in total, so they induce an additive constant of at most $2$ and thus can be ignored for the analysis.
Each different type $\binType\in\binSet:=\{\binBL,\binBS, \binB, \binLLS, \binLL, \binLSS, \binSSS\}$ is given a priority
such that we have a total ordering $<$ on $\binSet$ which is
$\binBL>\binBS> \binB> \binLLS> \binLL> \binLSS >\binSSS$.
Out of the listed bin types above, the MP algorithm utilizes all but bins of type $\binLSS$.

The algorithm works in a myopic manner:
If item $a$ has to be inserted into the packing, the algorithm disregards all items with a type lower (w.r.t.\ the range of size) than the type of $a$ during the insertion, i.e., it acts as if those smaller items would not exist.
It now inserts the item in the first fitting bin regarding the given priority of bin types.
All items with a lower type in this bin will be removed to an auxiliary storage and afterwards inserted in a recursive manner.
For a deletion, the item is removed and all other items in the same bin are moved to the auxiliary storage and inserted again with the same procedure\footnote{A detailed description of the MP algorithm can be found in Appendix~\ref{appendix:mp} and in~\cite{Ivkovic1996}.}.


\paragraph*{Properties of the Algorithm}
\label{sec:JohnsonClaims}
For our analysis of the packing,
we mainly use the thoroughness property which is ensured by the algorithm and proven in~\cite{Ivkovic1996}:

\begin{lemma}[\cite{Ivkovic1996}]\label{lemma:thoroughness}
A bin of type $\binType\in\binSet$ is thorough if there do not exist two bins $B_1$ and $B_2$
with lower types $\binType_1,\binType_2 \in \binSet$, $\binType_1,\binType_2 < \binType$ such that items from $B_1$ and $B_2$
can be used to form a bin of type $\binType$.
In the solution of the MP algorithm, bins of type $\binBL$, $\binBS$ and $\binLLS$
are thorough.
\end{lemma}

We denote by $\NBA{\binType}$ and $\NBO{\binType}$ (e.g. $\NBA{\binBL}$ and $\NBO{\binBL}$) the number of bins of type $\binType\in\binSet$ in $\ALG$ and $\OPT$, respectively.
The same notation is adapted for multiple types of bins (e.g. $\NBA{\binBL,\binLLS}=\NBA{\binBL}+\NBA{\binLLS}$).
Using this notation, we get:\\
$
\ALG \leq \NBA{\binBL} + \NBA{\binBS} + \NBA{\binB} + \NBA{\binLLS} + \NBA{\binLL} + \NBA{\binSSS} + 2,
$ and \\
$
\OPT \geq \NBO{\binBL} + \NBO{\binBS} + \NBO{\binB} + \NBO{\binLLS} + \NBO{\binLL} + \NBO{\binSSS} +\NBO{\binLSS}.
$

Note that the only reason these are not  equalities is that the number of bins of type $\binLS,\binL,\binSS$ or $\binS$ is between $0$ and $2$. 

We first argue that we may assume that no $\itemO$ items are part of the input:
Consider the case that there is a bin containing only $\itemO$ items in $\ALG$.
Then every bin in $\ALG$ except one is filled with items of a cumulative size of at least $\nicefrac{3}{4}$.
Together with Lemma~\ref{lemma:smallmain} we then directly get an approximation factor of $(1+\varepsilon)\alpha$
even if we would not combine our solutions for minor and major items at all.
Regarding the case that there is no such $\itemO$ bin,
since MP packs items in a myopic manner, it would produce the same solution if the $\itemO$ items were
not part of the input.
Hence we may compare the solution to an instance of the optimal solution, which does not need to consider any $\itemO$ items
in its instance.

In order to ease the analysis of the algorithm,
we introduce assumptions to the optimal solution and show
that these do not increase the value of the optimal solution.

\begin{lemma}
\label{lemma:optstructure}
Let $\ALG$ be the packing of the MP algorithm of a set of major items.
For an optimal solution $\OPT$ of a packing of the same items, the following properties can be assumed without increasing the number of used bins in $\OPT$:
\begin{compactenum}
  \item $\OPT$ does not pack a $\binBS$ bin containing a $\itemB$ item that is part of a $\binBL$ bin in $\ALG$.
	\item $\OPT$ does not pack a $\binB$ bin containing a $\itemB$ item that is part of a $\binBL$ bin in $\ALG$.
\end{compactenum}
\end{lemma}

We use the thoroughness of the MP algorithm (cf.~Lemma~\ref{lemma:thoroughness}) to show the following four statements that compare the solutions of $\ALG$ and $\OPT$ with each other. 
They will later be used in the analysis of the combination of our two approaches.

\begin{lemma}
	\label{claim:boundB}
	$
	\NBA{\binBL}+\NBA{\binBS}+\NBA{\binB}=\NBO{\binBL}+\NBO{\binBS}+\NBO{\binB}
	$
\end{lemma}

\begin{lemma}
	\label{claim:boundL}
	$
	\NBA{\binBL}+\NBO{\binLSS} \geq 2 \left(\NBA{\binLLS,\binLL} - \NBO{\binLLS, \binLL} \right)
	$
\end{lemma}

\begin{lemma}
	\label{claim:boundS}
	$
	\NBA{\binBS} + \NBA{\binLL} + 2\NBO{\binLSS} \geq 3 \left(\NBA{\binSSS} - \NBO{\binSSS} \right)
	$
\end{lemma}

\begin{lemma}
	\label{claim:boundS2}
	$
	\NBA{\binBS} + \NBO{\binLLS} + 2\NBO{\binLSS}\geq 3 \left(\NBA{\binSSS} - \NBO{\binSSS} \right)
	$
\end{lemma}

We finally determine the number of shifting moves that occur during the insertion or deletion of a major item.
Note that while considering these, we need to also account for possible $\itemO$ items.

\begin{lemma}\label{lemma:large_shifting_moves}
	The number of shifting moves (regarding major items) is bounded by $\mathcal{O}(\nicefrac{1}{\varepsilon})$.
\end{lemma}


\section{Combining Major and Minor Items}\label{sec:Combination}

The two presented approaches for the packing of minor and major items treat these items independently.
These independent solutions now have to be combined into one to reach a good approximation guarantee.
Assume we obtain these solutions as two sets of bins, where the bins $\left(B_1^{\text{min}},\ldots,B_m^{\text{min}}\right)$ are the result of the algorithm for minor items and the bins $\left(B_1^{\text{maj}},\ldots,B_n^{\text{maj}}\right)$ are the result of the algorithm for major items.
Note that although we temporarily ignored items of a size in $\left(\nicefrac{\varepsilon}{15},\nicefrac{1}{4}\right]$ in the analysis for the major items, the algorithm that
combines the two solutions of course does not ignore the potential $\itemO$ items.

We first describe the structure of the packing we want to achieve and then show how to maintain that structure over time.
The goal is to create pairs of bins with one bin from each solution while not modifying too many pairs in each time step.
For ease of description, we still refer to two bins $B_i^{\text{min}}$ and $B_j^{\text{maj}}$ as two different bins even though their contents may be packed into the same bin.
In such a case, we say that $B_i^{\text{min}}$ is paired with $B_j^{\text{maj}}$.

We want to maintain a greedy-style combination of the two lists of bins, which can be described by the following combination process:
The list of bins with minor items is (partially) sorted by their type ($1,\ldots,k$), i.e. bins potentially filled with more minor items appear earlier (the ordering here is different compared to the ordering in Section~\ref{sec:small} when the minor items are actually packed).
The list of bins with major items is sorted by their filling height in decreasing order (regardless of their type).
The process iterates over the $k$ bin types in the solution of minor items starting with type 2 (since there is no reserved space in bins of type 1) in increasing order.
For each bin $B_i^{\text{min}}$ of type $j$, we iterate over the bins with major items starting with the bins that have the largest filling height.
We pair $B_i^{\text{min}}$ with a bin with major items $B_\ell^{maj}$ that has a filling height of at most $1-w_j$ and for which $\ell$ is minimal, i.e., the first bin with major items whose items fit into the reserved space of the respective (minor item) bin type.

Note that this process incorporates all bins of the minor solution that contain at least one item, including the ones that are part of a buffer group.
We do not use bins that contain no minor items at all (even if they are already present in the minor algorithm as part of a buffer group).
Such a greedy-style packing can be maintained while only modifying $\mathcal{O}(k)$ pairings per changed bin in either one of the two solutions.
A major reason for this is that for the bins with minor items, only their type is of interest.
The pairings of bins need to be changed if one of the following happens:

\noindent\textbf{A change in the solution of major items:}
For each insertion or deletion of a major item, the solution of major items is modified independently first.
The above described greedy process is then used to determine which bins with major items need to be matched with which types of bins of minor items.
The combination is then modified to fit the new solution by switching out the major bins where needed, starting with those which are paired with bins of minor items of type 2.

\begin{figure}[htb]
	\begin{minipage}[c]{0.5\textwidth}
		\includegraphics[scale=.33]{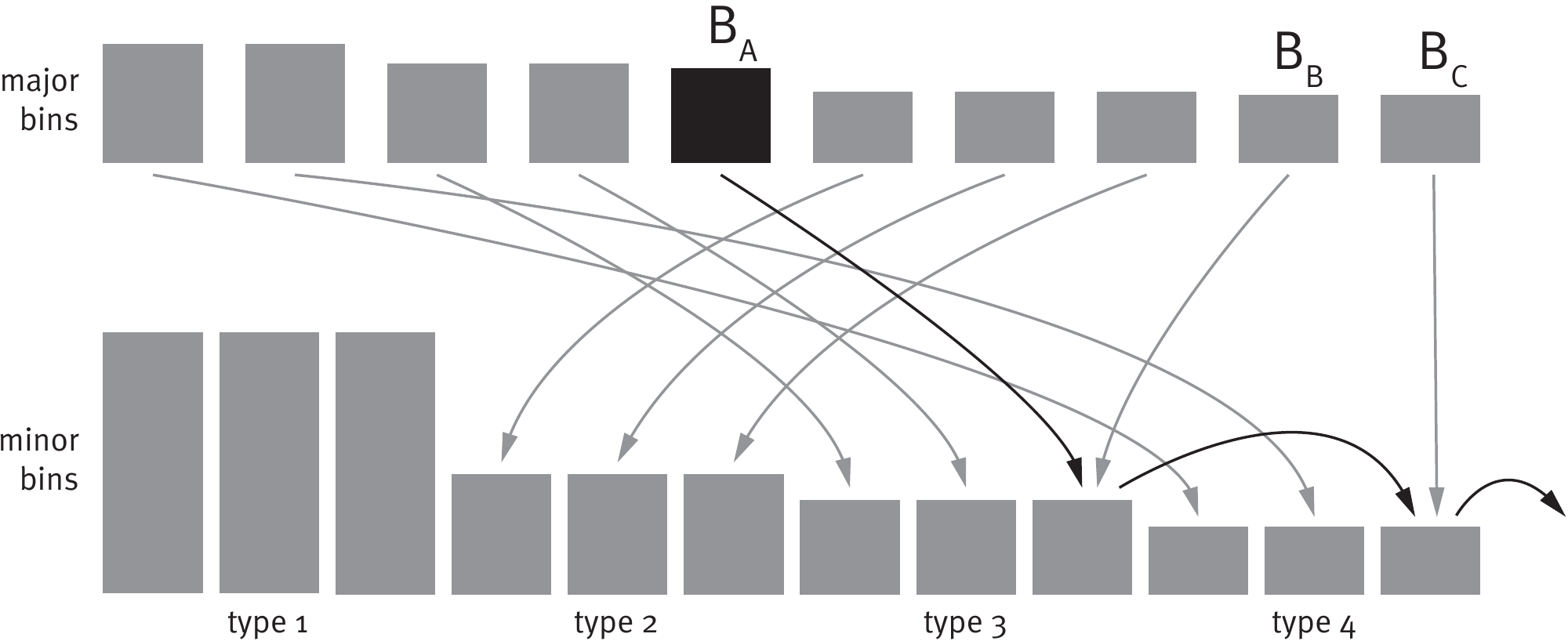}
	\end{minipage}\hfill
	\begin{minipage}[c]{0.5\textwidth}
		\caption{The gray bins represent current bins of the two solutions, the arrows indicate the current combination.
			The bin $B_A$ is now inserted into the solution for the major items.
			$B_A$ only fits into minor bins of type 3 or higher.
			The black arrows indicate the switching process. Bin $B_A$ displaces $B_B$ to a minor bin of type 4.
			Since there is no bin left for $B_C$, it is not combined with any minor bin after the changes.}
		\label{fig:algorithm_combination}
	\end{minipage}
\end{figure}

\noindent\textbf{A change in the solution of minor items:}
As for the major items, the solution of minor items is first modified independently upon insertion or deletion of a minor item.
The modification may add or remove at most one bin of minor items.
In this case, the above greedy approach is used to recalculate which bins with major items need to be matched with which types of bins of minor items.
The solution is modified accordingly, starting with the matching with bins with minor items of type 2.


\paragraph*{Analysis}
Due to the described greedy approach for changes and Lemma~\ref{lemma:small_shifting_moves} and~\ref{lemma:large_shifting_moves}, we can bound the total number of shifting moves.
\begin{lemma}\label{lemma:shifting_moves}
	The algorithm uses at most $\mathcal{O}(\nicefrac{1}{\varepsilon^2})$ shifting moves for each insertion or deletion.
\end{lemma}

\begin{lemma}
\label{lemma:combstructure}
Let $B^{\text{min}}$ be a bin with minor items of type $j$.
If $B^{\text{min}}$ is not combined with a bin of major items, all non-combined bins with
major items have a filling height of at least $1-w_j$.
\end{lemma}

In the remainder of this paper, we show the approximation quality our algorithm achieves.

\begin{lemma}\label{lemma:solution_quality}
	Let $\ALG$ be the number of bins our algorithm uses for an arbitrary input sequence and $\OPT$ the number of bins used by an optimal solution. Then this yields
	$\ALG \leq (1+\varepsilon) \cdot \alpha \cdot \OPT.$
\end{lemma}

\begin{proof}
We reuse the notation of Sections~\ref{sec:small} and~\ref{sec:large}.
For the bins with major items of $\ALG$, we introduce a collection of bins called \emph{\itemL-\itemS-quartet}, consisting of three $\binLL$ bins and one $\binSSS$ bin.
The number of such collections is denoted by $\quartet:=\left\lfloor\min\{\NBA{\binSSS},\frac{1}{3}\NBA{\binLL}\}\right\rfloor$.
We split the number of $\binLL$ and $\binSSS$ bins in $\ALG$ in the number of bins that can be put in such quartets, denoted by $\quartetLL:=3\quartet$ and $\quartetSSS:=\quartet$,
and the respective number of bins that cannot be put in such a collection, denoted by $\nonquartetLL$ and $\nonquartetSSS$.
Note that it therefore holds $\NBA{\binLL}=\quartetLL + \nonquartetLL$
and $\NBA{\binSSS}=\quartetSSS + \nonquartetSSS$.
Furthermore, due to the definition of $\quartet$, it holds $\nonquartetLL\leq 2$ or $\nonquartetSSS=0$.

Note that we assume that the solution of $\ALG$ for the major items does not contain bins with only $\itemO$ items, otherwise the approximation ratio would follow directly
as argued in the previous section.

We estimate the optimal solution by adding up all the items which need to be packed.
As before, we denote by $\workloadSmall$ the cumulative size of all minor items.
We estimate the cumulative size of major items by considering the different types of bins with major items $\binType\in\binSet$
and their minimal filling height $F(T)$ resulting from the minimum size of the respective items (e.g., for bins of type $\binBL$, we have $F(\binBL) = \nicefrac{1}{2} + \nicefrac{1}{3} = \nicefrac{5}{6}$).
By incorporating the fact that all bins in quartets have an average filling height of $\nicefrac{3}{4}$ we get
\begin{tightalign}
\OPT&\geq \workloadSmall + \sum_{\binVar\in\binSet} F(\binVar) \cdot \NBA{\binVar}^{-Q} + \frac{3}{4}\quartetLL + \frac{3}{4}\quartetSSS\text{ as well as } \label{eqn:optBound1}\\
\OPT&\geq \workloadSmall + \sum_{\binVar\in\binSet} F(\binVar) \cdot \NBA{\binVar}^{-Q} + 3\quartet. \label{eqn:optBound2}
\end{tightalign}

From Lemma~\ref{lemma:combstructure} we get the following:
Suppose $j$ is the maximum index such that a bin of type $j$ is not combined with a bin of major items
(pick $j=1$ if such a bin does not exist).
Then all remaining bins with major items must have a size of at least $1-w_j$.
Note that if $j=k$, then $\ALG \leq (z_k'+\frac{3}{4}\varepsilon\alpha) \workloadSmall+z_k'\cdot W^{maj}$ directly follows, where $W^{maj}$ is the workload of major items, and hence the approximation ratio. We assume $j<k$ in the following.
We introduce $\NBA{\binBL,\binBS, \binB, \binLLS, \binLL, \binSSS}^{\geq(1-w_j);-Q}$ as the number of bins (with major items) of the given types used by our algorithm, limited to bins with a filling height of at least $1-w_j$ and excluding all quartets.
Hence we have
\begin{tightalign}
  \ALG &\leq (z_j'+\frac{3}{4}\varepsilon\alpha) \workloadSmall + \NBA{\binBL,\binBS, \binB, \binLLS, \binLL, \binSSS}^{\geq(1-w_j);-Q} + \quartetLLSSS \nonumber \\
	  &\leq (z_j'+\frac{3}{4}\varepsilon\alpha) \OPT+ \sum_{\binVar\in\binSet}(1-z_j'\cdot F(\binVar))\cdot \NBA{\binVar}^{\geq(1-w_j);-Q} + (1-\frac{3}{4}z_j')\quartetLLSSS \label{eqn:optBound1Improved}
\end{tightalign}
from (\ref{eqn:optBound1}) as well as
\begin{tightalign}
  \ALG &\leq (z_j'+\frac{3}{4}\varepsilon\alpha) \workloadSmall + \NBA{\binBL,\binBS, \binB, \binLLS, \binLL, \binSSS}^{\geq(1-w_j);-Q} + 4\quartet \nonumber \\
	  &\leq (z_j'+\frac{3}{4}\varepsilon\alpha) \OPT+ \sum_{\binVar\in\binSet}(1-z_j'\cdot F(\binVar))\cdot \NBA{\binVar}^{\geq(1-w_j);-Q} + (4-3z_j')\quartet. \label{eqn:optBound2Improved}
\end{tightalign}
from (\ref{eqn:optBound2}).
We use one of these estimations depending on $\quartet$.

Let $\hat{z}_j:=z_k'-z_j'$. For the two cases with $\nonquartetLL\leq 2$ or $\nonquartetSSS=0$ we show the following lemmas:
\begin{lemma}\label{lemma:OPT_LL}
	Let $\nonquartetLL\leq 2$. Then it holds that
\begin{tightalign*}
 \OPT\geq\frac{1}{\hat{z}_j}\cdot\left(\sum_{\binVar\in\binSet}(1-z_j'\cdot F(\binVar))\cdot \NBA{\binVar}^{\geq(1-w_j);-Q} + (1-\frac{3}{4}z_j')\quartetLLSSS\right).
\end{tightalign*}
\end{lemma}
\begin{lemma}\label{lemma:OPT_SSS}
	Let $\nonquartetSSS=0$. Then it holds that
\begin{tightalign*}
 \OPT\geq\frac{1}{\hat{z}_j}\cdot\left(\sum_{\binVar\in\binSet}(1-z_j'\cdot F(\binVar))\cdot \NBA{\binVar}^{\geq(1-w_j);-Q} + (4-3z_j')\quartet\right).
\end{tightalign*}
\end{lemma}
Hence, from (\ref{eqn:optBound1Improved}) together with Lemma~\ref{lemma:OPT_LL} or (\ref{eqn:optBound2Improved}) together with Lemma~\ref{lemma:OPT_SSS}, we conclude
\begin{tightalign*}
 \ALG \leq (z_j' + \frac{3}{4}\varepsilon\alpha)\OPT + \hat{z}_j\cdot \OPT \leq (z_k' + \frac{3}{4}\varepsilon\alpha)\OPT = (1+\varepsilon)\alpha\OPT.
\end{tightalign*}

\vspace*{-1.8em}
\end{proof}

\noindent
Finally, Theorem~\ref{theorem:main_result} now directly follows from our analysis: Lemma~\ref{lemma:solution_quality} gives the approximation ratio of $(1+\varepsilon)\cdot \alpha$ and Lemma~\ref{lemma:shifting_moves} bounds the number of shifting moves to $\mathcal{O}(\nicefrac{1}{\varepsilon^2})$.

The lower bound of Balogh~et~al.~\cite{DBLP:journals/siamcomp/BaloghBGR08} also applies for the Relaxed Online Bin Packing Problem, thus we can also close the gap between upper and lower bound for this problem with the following direct corollary.
\begin{corollary}\label{cor:relaxedOnlineBinPacking}
For each $\varepsilon \in (0,1)$, there exists an algorithm for the Relaxed Online Bin Packing Problem with
an asymptotic approximation ratio of $\left(1+\varepsilon\right)\cdot\left(1-\nicefrac{1}{\left(W_{-1}\left(\nicefrac{-2}{e^3}\right) + 1\right)}\right)$
which repacks at most $\mathcal{O}\left(\nicefrac{1}{\varepsilon^2}\right)$ items per insertion of an item.
\end{corollary}


\newpage
\bibliographystyle{plain}
\bibliography{references}



\newpage
\appendix
\allowdisplaybreaks
\section{Proofs from Section~\ref{sec:small}}

{
	\renewcommand{\thetheorem}{\ref{lemma:ordering_feasibility}}
	
	\begin{lemma}
	  During the whole execution of the algorithm it holds for two arbitrary bins $B_i$ and $B_{i'}$ with $i<i'$ that $\min_{e \in B_i}e \geq \max_{e \in B_{i'}}e$.
	  Furthermore, for each bin $B$ with type $j$ it holds that $\sum_{e\in B}e\leq w_j$.
	\end{lemma}
	
}
\begin{proof}
	\emph{Ordering of elements:}
	Consider a given packing which abides the described ordering.
	When an item $a$ is inserted into a bin $B_i$, it is ensured that all items in $B_{i-1}$ are at least as large as $a$ and all items in $B_{i+1}$ are at most as large as $a$.
	The recursive call always either shifts the largest item from $B_i$ into $B_{i-1}$ or the smallest item from $B_{i-1}$ into $B_i$ and therefore upholds the ordering.
	
	\emph{Feasibility:}
	For the deletion of an item, the statement holds since an item is only shifted from $B_{i-1}$ into $B_i$ if it does not violate
	the feasibility condition.
	For the insertion procedure, we observe that when an item is added to a bin $B_i$, it is sufficient to remove at most the largest item in $B_i$
	to restore the feasibility condition for this bin.
\end{proof}

{
	\renewcommand{\thetheorem}{\ref{lemma:distance}}
	
	\begin{lemma}
		The number of full groups between two buffer groups is in $[l,2l]$.
	\end{lemma}

}
\begin{proof}
	During insertion, new buffer groups are only inserted in case 3b.
	Inserting a new buffer group to the left of $B_i$ maintains the number of full groups to the left of $B_i$ and increases the number of full groups to the right by 1
	since the group of $B_i$ is now full.
	If this number now becomes $2l$ it is obvious that there is a position for a new buffer group to the right of the group of $B_i$, such that there are always at least $l$
	full groups between two buffer groups.
	
	Regarding the deletion process, only in case 3b the buffer groups are changed.
	Removing the group to the left of $B_i$ and declaring the group of $B_i$ a buffer group does not change
	the number of full groups to the left of $B_i$ and decreases the number of full groups to the right by 1.
	Hence, only when this number drops below $l$, further changes need to be made.
	The recursive deletion makes the group of $B_i$ full again.
	If afterwards the distance between two buffer groups is at least $2l$, the insertion of a new group is possible with respect to the constraints.
\end{proof}

{
	\renewcommand{\thetheorem}{\ref{lemma:fitting}}
	
	\begin{lemma}
		A set of full bins which consists of at least $z_j\cdot\workloadSmall$ bins of type $j$ for each bin type $j$ contains a workload of at least $\workloadSmall$ of minor items,
i.e. $\sum_{j=1}^k z_j\workloadSmall \cdot (w_j - \delta) \geq \workloadSmall$.
	\end{lemma}

}
\begin{proof}
	Note that this proof explicitly requires $z_k'$ to be larger than $\alpha$. Intuitively, this proof implies that no feasible parameters for our algorithm can be found
	if it wants to pack minor items in less than $\alpha\cdot\workloadSmall$ bins.
	Starting with the statement, the definition of the filling height with $w_1 =1$ and $w_i = y_i\ \forall i>1$ and using the fact that $z_i = z'_i - z'_{i-1}$,
	the following are equivalent:
	\begin{align*}
	\sum_{i=1}^k z_i \cdot (w_i - \delta) &\geq 1 \\
	z'_1 + \sum_{i=2}^kz_iy_i - \delta\sum_{i=1}^kz_i &\geq 1\\
	z'_1 + \sum_{i=2}^k(z'_i -z'_{i-1})y_i - \delta z'_k &\geq 1
	\end{align*}
	
	Using the definition of $z'_i = \frac{y_k}{y_i(1-y_k)}$, we get:
	\begin{align*}
	\frac{y_k}{y_1(1-y_k)} + \sum_{i=2}^k\left(\frac{y_k}{y_i(1-y_k)}-\frac{y_k}{y_{i-1}(1-y_k)}\right)y_i &\geq 1+ \delta z'_k \\
	\frac{y_k}{(1-y_k)}\left(\frac{1}{y_1} + \sum_{i=2}^k\left(1-\frac{y_i}{y_{i-1}}\right)\right)  &\geq 1 + \delta z'_k
	\end{align*}
	
	We now use the definition of $y_i = \frac{1}{2}(2y_k)^{\frac{i-1}{k-1}}$ and $z'_k = \frac{1}{1-y_k}$:
	\begin{align*}
	z'_ky_k\left(2 + \sum_{i=2}^k\left(1-\frac{\frac{1}{2}(2y_k)^\frac{i-1}{k-1}}{\frac{1}{2}(2y_k)^\frac{i-2}{k-1}}\right) \right) &\geq 1 + \delta z'_k\\
	2 + \sum_{i=2}^k\left(1-(2y_k)^\frac{1}{k-1}\right) &\geq \frac{1}{z'_ky_k} + \frac{\delta}{y_k}\\
	2+ (k-1)\left(1-(2y_k)^\frac{1}{k-1}\right) &\geq \frac{1-y_k}{y_k} + \frac{\delta}{y_k}\\
	(k-1)\left(1- e^{\ln{(2y_k)}\cdot \frac{1}{k-1}}\right) &\geq \frac{1+\delta-3y_k}{y_k} 
	\end{align*}
	
	In the following, we apply the exponential series $e^x = \sum_{j=0}^{\infty}\frac{x^j}{j!}$:
	\begin{align}
	(k-1)\left(1- \sum_{j=0}^\infty \left(\frac{\ln{2y_k}}{k-1}\right)^j/j!\right) &\geq \frac{1+\delta-3y_k}{y_k} \nonumber\\
	(k-1) - (k-1) - (k-1)\sum_{j=1}^\infty \left(\frac{\ln{2y_k}}{k-1}\right)^j/j! &\geq \frac{1+\delta-3y_k}{y_k} \nonumber\\
	- (k-1)\sum_{j=1}^\infty \frac{\ln{2y_k}}{k-1}\left(\frac{\ln{2y_k}}{k-1}\right)^{j-1}/j! &\geq \frac{1+\delta-3y_k}{y_k} \nonumber\\
	- \ln(2y_k)\sum_{j=1}^\infty \left(\frac{\ln{2y_k}}{k-1}\right)^{j-1}/j! &\geq \frac{1+\delta-3y_k}{y_k} \nonumber\\
	- \ln(2y_k) - \frac{\ln^2(2y_k)}{2(k-1)} - \ln(2y_k)\sum_{j=3}^\infty \left(\frac{\ln{2y_k}}{k-1}\right)^{j-1}/j! &\geq \frac{1+\delta-3y_k}{y_k}\label{eq:exp_inequality}
	\end{align}
		
	We first show that $- \ln(2y_k)\sum_{j=3}^\infty \left(\frac{\ln{2y_k}}{k-1}\right)^{j-1}/j! \geq 0$ with the following transformation:
	\begin{align}
	&- \ln(2y_k)\sum_{j=3}^\infty \frac{1}{j!} \left(\frac{\ln{2y_k}}{k-1}\right)^{j-1} \nonumber\\
	=&- (k-1) \sum_{j=3}^\infty \frac{1}{j!} \frac{\ln{2y_k}}{k-1}\left(\frac{\ln{2y_k}}{k-1}\right)^{j-1}\nonumber\\
	=&- (k-1)\sum_{j=3}^\infty \frac{1}{j!} \left(\frac{\ln{2y_k}}{k-1}\right)^{j} \nonumber\\
	=&- (k-1)\sum_{j=1}^\infty \frac{1}{(2j+1)!} \left(\frac{\ln{2y_k}}{k-1}\right)^{2j+1} + \frac{1}{(2j+2)!} \left(\frac{\ln{2y_k}}{k-1}\right)^{2j+2} \nonumber\\
	=&- (k-1)\sum_{j=1}^\infty \frac{1}{(2j+1)!} \left(\frac{\ln{2y_k}}{k-1}\right)^{2j+1}\left(1+\frac{\ln(2y_k)}{(k-1)(2j+2)}\right) \label{eq:exp_greater_0}
	\end{align}
	
	The last two steps follow by separating even and odd summands. In~\eqref{eq:exp_greater_0}, the term $\frac{1}{(2j+1)!} \left(\frac{\ln{2y_k}}{k-1}\right)^{2j+1}$ is always negative since $y_k =1-\nicefrac{1}{z'_k}$. $\frac{\ln(2y_k)}{(k-1)(2j+2)}$ is always between $-1$ and $0$. Together with the prior negative factor the statement follows.

From~\eqref{eq:exp_greater_0} we get that~\eqref{eq:exp_inequality} is implied by the following equivalent statements.
We use $y_k= \frac{z'_k-1}{z'_k}$ for the transformation:
	
	\begin{align}
	- \ln(2y_k) - \frac{\ln^2(2y_k)}{2(k-1)} &\geq \frac{1+\delta-3y_k}{y_k}\nonumber\\
	- \frac{\ln^2(2y_k)}{2(k-1)} &\geq \frac{1+\delta-3y_k+y_k\ln(2y_k) }{y_k} \nonumber\\
	- \frac{\ln^2(2y_k)}{2(k-1)} &\geq \frac{1+\delta-3\frac{z_k'-1}{z_k'}+\frac{z_k'-1}{z_k'}\ln(2\frac{z_k'-1}{z_k'})}{\frac{z_k'-1}{z_k'}}\nonumber\\
	- \frac{\ln^2(2y_k)}{2(k-1)} &\geq \frac{z_k'+\delta z_k'-3(z_k'-1)+(z_k'-1)\ln(2\frac{z_k'-1}{z_k'})}{z_k'-1}\nonumber\\
	- \frac{\ln^2(2y_k)}{2(k-1)} &\geq \frac{-2z_k'+\delta z_k'+3+(z_k'-1)\ln(2\frac{z_k'-1}{z_k'})}{z_k'-1}\nonumber
\end{align}
	
We choose a new $\gamma<0$ such that $\ln(2\frac{z_k'-1}{z_k'})=(1+\gamma)\ln(2\frac{\alpha-1}{\alpha})$ and substitute it in the term:
\begin{align}
	- \frac{\ln^2(2y_k)}{2(k-1)} &\geq \frac{-2z_k'+\delta z_k'+3+(z_k'-1)(1+\gamma)\ln(2\frac{\alpha-1}{\alpha})}{z_k'-1}\nonumber\\
	- \frac{\ln^2(2y_k)}{2(k-1)} &\geq \frac{-2z_k'+\delta z_k' +3 +(z_k'-1)\ln(2\frac{\alpha-1}{\alpha}) + \gamma(z_k'-1)\ln(2\frac{\alpha-1}{\alpha})}{z_k'-1}\nonumber
\end{align}

Applying our definition of $z_k'=\left(1+\frac{\varepsilon}{4}\right)\alpha$ yields to:
\begin{align}
- \frac{\ln^2(2y_k)}{2(k-1)} \geq&\frac{1}{z_k'-1}\cdot\left( -2\left(1+\frac{\varepsilon}{4}\right)\alpha+\delta \left(1+\frac{\varepsilon}{4}\right)\alpha +3 +(\left(1+\frac{\varepsilon}{4}\right)\alpha-1)\ln(2\frac{\alpha-1}{\alpha})\right.\nonumber \\
&\left.+ \gamma(z_k'-1)\ln(2\frac{\alpha-1}{\alpha}){z_k'-1}\right)\nonumber\\
- \frac{\ln^2(2y_k)}{2(k-1)}
\geq &\frac{-2\alpha +3 +(\alpha-1)\ln(2\frac{\alpha-1}{\alpha})}{z_k'-1}\nonumber \\
&+\frac{-\frac{\varepsilon}{2}\alpha + \delta (1+\frac{\varepsilon}{4})\alpha + \frac{\epsilon}{4}\alpha \ln(2\frac{\alpha-1}{\alpha}) + \gamma(z_k'-1)\ln(2\frac{\alpha-1}{\alpha})}{z_k'-1}\label{eq:two_alpha_summands}
\end{align}

The last step only splits the term in two summands. We first concentrate on the first one and show that $-2\alpha +3 +(\alpha-1)\ln(2\frac{\alpha-1}{\alpha})=0$.
We first use the definition of $\alpha = 1-\nicefrac{1}{\left(W_{-1}\left(\frac{-2}{e^3}\right)+1\right)}$ and then the fact that $\ln\left(\frac{-2}{W_{-1}(\frac{-2}{e^3})}\right)=W_{-1}\left(\frac{-2}{e^3}\right)+3$.

\begin{align}
	&-2\alpha +3 +(\alpha-1)\ln(2\frac{\alpha-1}{\alpha})\nonumber\\
	=&-2\alpha +3 +(\alpha-1)\ln\left(2\frac{-\frac{1}{\left(W_{-1}\left(\frac{-2}{e^3}\right)+1\right)}}{1-\frac{1}{\left(W_{-1}\left(\frac{-2}{e^3}\right)+1\right)}}\right)\nonumber\\
	=& -2\alpha +3 +(\alpha-1) \left(W_{-1}\left(\frac{-2}{e^3}\right)+3\right)\nonumber\\
	=& -2\left(1-\frac{1}{\left(W_{-1}\left(\frac{-2}{e^3}\right)+1\right)}\right) +3\nonumber\\
	&+\left(\left(1-\frac{1}{\left(W_{-1}\left(\frac{-2}{e^3}\right)+1\right)}\right)-1\right)  \left(\left(W_{-1}\left(\frac{-2}{e^3}\right)+1\right) +2\right)\nonumber\\
	=& -2
	+ \frac{2}{\left(W_{-1}\left(\frac{-2}{e^3}\right)+1\right)}
	+3
	-\frac{W_{-1}\left(\frac{-2}{e^3}\right)+1}{W_{-1}\left(\frac{-2}{e^3}\right)+1}
	-\frac{2}{W_{-1}\left(\frac{-2}{e^3}\right)+1}\nonumber\\
	=&\ 0\label{eq:W_term_0}
\end{align}

Additionally, we show a further bound for $\gamma\ln(2\frac{\alpha-1}{\alpha})$, using $\alpha>1.375$.

\begin{align}
  \gamma\ln(2\frac{\alpha-1}{\alpha}) &= \ln(2\frac{z_k'-1}{z_k'}) - \ln(2\frac{\alpha-1}{\alpha}) 
	  = \ln\left(\frac{(1+\frac{\varepsilon}{4})\alpha-1}{(1+\frac{\varepsilon}{4})(\alpha-1)}\right)
		\leq \ln(1+\frac{2}{3}\varepsilon) 
		\leq \frac{2}{3}\varepsilon\label{eq:bound_gamma_ln}
\end{align}

Hence,~\eqref{eq:two_alpha_summands} is implied by the following, using~\eqref{eq:W_term_0} to eliminate the left summand and~\eqref{eq:bound_gamma_ln}:
\begin{align*}
	- \frac{\ln^2(2y_k)}{2(k-1)} &\geq \frac{-\frac{\varepsilon}{2}\alpha+\delta(1+\frac{\varepsilon}{4})\alpha
	   +\frac{\varepsilon}{4}\alpha\ln(2\frac{\alpha-1}{\alpha})+\frac{2}{3}\varepsilon(z_k'-1)}{z_k'-1}
\end{align*}

Since $\delta=\frac{\varepsilon}{15}$ implies $-\frac{\varepsilon}{2}\alpha+\delta(1+\frac{\varepsilon}{4})\alpha
	   +\frac{2}{3}\varepsilon(z_k'-1)\leq 0$, we get to
\begin{align*}
	- \frac{\ln^2(2y_k)}{2(k-1)} &\geq \frac{\frac{\varepsilon}{4}\alpha\ln(2\frac{\alpha-1}{\alpha})}{z_k'-1} \\
	k &\geq 1 + \frac{1}{\varepsilon}\frac{-\ln^2(2y_k)(z_k'-1)}{\frac{1}{2}\alpha\ln(2\frac{\alpha-1}{\alpha})} \\
	k &\geq 1 + \frac{3}{\varepsilon}.
\end{align*}

\end{proof}

{
	\renewcommand{\thetheorem}{\ref{lemma:smallmain}}
	
	\begin{lemma}
		For all $1\leq j\leq k$, the total number of bins of any type $i\leq j$ is at most $(z_j'+\frac{3}{4}\varepsilon\alpha)\workloadSmall$.
	\end{lemma}
	
}
\begin{proof}
	The number of bins of type $\leq j$ in a group of size $l$ is set to
	$\left\lceil\frac{z_j'}{z_k'}\cdot l\right\rceil \leq \frac{z_j'}{z_k'}\cdot l+1 \leq \left(\frac{z_j'}{z_k'}+\frac{\epsilon\alpha}{4z_k'}\right)\cdot l$.
	Note that obviously $\left\lceil\frac{z_j'}{z_k'}\cdot l\right\rceil\geq \frac{z_j'}{z_k'}\cdot l$.
	
	We estimate the cumulative size of minor items a full group (consisting of $l$ bins) contains.
	As noted above, the rounding is done in favor of the bin types which contain more items,
	hence a full group holds minor items of cumulative size of at least
	$\sum_{j=1}^{k}\frac{z_j}{z_k'}\cdot l \cdot(w_j-\delta)\geq\frac{l}{z_k'}$ (cf. Lemma~\ref{lemma:fitting}).
	As a consequence, in order to pack all minor items at most $\frac{\workloadSmall\cdot z_k'}{l}$ full groups are necessary.
	
	To account for the buffer groups, we divide all bin groups into units of one buffer group and at least $l$ full
	groups which are located to the right of it in the ordering of the bins.
	At most $\frac{\workloadSmall\cdot z_k'}{l^2}$ such units are necessary to host all minor items.
	The number of bins of type $\leq j$ in such a unit are at most $(l+1)\cdot \left(\frac{z_j'}{z_k'}+\frac{\epsilon\alpha}{4z_k'}\right)\cdot l$.
	It follows that the total number of bins of type $\leq j$ are at most
	$\workloadSmall\cdot z_k'\cdot\frac{(l+1)l}{l^2}\cdot \left(\frac{z_j'}{z_k'}+\frac{\epsilon\alpha}{4z_k'}\right)\leq \left(1+\frac{\epsilon\alpha}{4z_k'}\right)\left(z_j'+\frac{\epsilon\alpha}{4}\right)\workloadSmall
	\leq\left(z_j'+\frac{3}{4}\varepsilon\alpha\right)\workloadSmall$.
\end{proof}

{
	\renewcommand{\thetheorem}{\ref{lemma:small_shifting_moves}}
	
	\begin{lemma}
		The number of shifting moves (regarding minor items) is bounded by $\mathcal{O}(\nicefrac{1}{\varepsilon^2})$.
	\end{lemma}
	
}
\begin{proof}
	For both insertions and deletions, it is obvious that the number of shifting moves is equal to the number of recursive calls during the procedure.
	
	\textit{Insertion:} The recursive calls in case 1 of the algorithm only occur if the current bin $B_i$ is a full bin, which can be at most $(2l+1)\cdot l$ times directly after another, since
	the left-most bin of a buffer group is never full.
	The recursive call in case 3 is applied to an empty bin, since either it is part of the same buffer group as $B_i$ or a bin of a newly created buffer group.
	The number of recursive calls is thereby bounded by $\mathcal{O}(l^2)$.
	
	\textit{Deletion:} The recursive call in case 2 occurs at most $(2l+1)\cdot l$ times since this is the maximum number of full bins in a row.
	The recursive call in case 3b occurs at most once, since after the removal of the buffer group next to $B_i$, the distance between the buffer groups to the left
	and to the right of the group of $B_i$ is at least $2l-1$.
	Hence for the deletion, the number of recursive calls is also bounded by $\mathcal{O}(l^2)$.
	
	Since $l=\mathcal{O}(\nicefrac{1}{\varepsilon})$, the Lemma directly follows.
\end{proof}


\section{Description of the MP Algorithm}\label{appendix:mp}

Since~\cite{Ivkovic1996} is not freely available on the Internet and the algorithm is one essential component of our algorithm,
we give a more detailed description of the MP algorithm in addition to the general idea given in Section~\ref{sec:large}.

In addition to the regular packing, the algorithm has an auxiliary storage to which items may be temporarily moved.
At the end of an insert or delete operation, only a constant number of items remains in the auxiliary storage and is packed into
at most 2 bins.
For an insertion of an item $a$, it is simply added to the auxiliary storage and a procedure to clear the storage is called.
For a deletion of $a$, all items in the same bin as $a$ are removed from the regular packing and added to the auxiliary storage.
Then the procedure to clear the storage is called.
This procedure works as follows:
\begin{compactenum}
	\item Every $\itemB$ item from the auxiliary storage is inserted into the regular packing. The thoroughness property is maintained by successively searching
	for fitting $\itemL$ and then $\itemS$ items in bins of a lower type to pair with the new $\itemB$ item. The remaining items from the bins from which the
	$\itemL$ or $\itemS$ item was removed are moved to the auxiliary storage.
	\item $\itemL$ and $\itemS$ items from the auxiliary storage are paired with $\itemB$ bins from the regular packing whenever possible. Bins of the same or higher type
	are not changed in the process, i.e. an $\itemL$ item can only be paired with a $\itemB$ item of a $\binBS$ or $\binB$ bin. Other items from the bin in which these items
	are inserted are moved to the auxiliary storage.
	\item As long as there are at least two $\itemL$ items in the auxiliary storage, they are inserted in a new bin and potentially paired with an $\itemS$ item either from
	the auxiliary storage or from an existing $\binSSS$ bin to form an $\binLLS$ bin. If no fitting $\itemS$ item exists, an $\binLL$ bin is created.
	If an $\itemS$ item is taken from a regular bin, the remaining items are moved to the auxiliary storage.
	\item As long as there are at least three $\itemS$ items in the auxiliary storage, new bins of type $\binSSS$ are formed with these items and inserted into the regular packing.
	At the end of this step, the auxiliary storage contains at most one $\itemL$ item and two $\itemS$ items which can be packed into at most two bins.
	\item All remaining $\itemO$ items in the auxiliary storage are moved to the regular packing in a first fit manner. This implies that bins that contain any other item type than
	$\itemO$ are prioritized over bins that exclusively contain $\itemO$ items.
\end{compactenum}


\section{Proofs from Section~\ref{sec:large}}

{
	\renewcommand{\thetheorem}{\ref{lemma:optstructure}}
	
	\begin{lemma}
		Let $\ALG$ be the packing of the MP algorithm of a set of major items.
		For an optimal solution $\OPT$ of a packing of the same items, the following properties can be assumed without increasing the number of used bins in $\OPT$:
		\begin{compactenum}
			\item $\OPT$ does not pack a $\binBS$ bin containing a $\itemB$ item which is part of a $\binBL$ bin in $\ALG$.
			\item $\OPT$ does not pack a $\binB$ bin containing a $\itemB$ item which is part of a $\binBL$ bin in $\ALG$.
		\end{compactenum}
	\end{lemma}

}
\begin{proof}We investigate both properties separately:
	\begin{compactenum}
		\item
		Assume the optimal solution places an $\itemS$ item $s$ together with a $\itemB$ item $b$ which is put together with an $\itemL$ item $l$ in a $\binBL$ bin in $\ALG$.
		Exchanging the positions of $s$ and $l$ in $\OPT$ yields a feasible solution with the same number of bins: Item $l$ fits together with $b$, because it is packed together in $\ALG$ and $size(s)<size(l)$ implies that $s$ fits into every bin in which $l$ was located before.
		
		\item
		Assume the optimal solution does have a $\binB$ bin with a $\itemB$ item $b$ which is combined with an $\itemL$ item $l$ in a $\binBL$ bin in $\ALG$.
		Placing $l$ together with $b$ does not increase the number of used bins in $\OPT$ since the two items are packed together in $\ALG$ and therefore they fit together in one bin.
	\end{compactenum}
\end{proof}

{
	\renewcommand{\thetheorem}{\ref{claim:boundB}}
	
	\begin{lemma}
		$
		\NBA{\binBL}+\NBA{\binBS}+\NBA{\binB}=\NBO{\binBL}+\NBO{\binBS}+\NBO{\binB}
		$
	\end{lemma}
	
}
\begin{proof}
	Two $\itemB$ items cannot be in the same bin by definition since their sizes are larger than $\frac{1}{2}$. Therefore, $\ALG$ and $\OPT$ have the same number of bins which involves a $\itemB$ item.
\end{proof}

{
	\renewcommand{\thetheorem}{\ref{claim:boundL}}
	
	\begin{lemma}
		$
		\NBA{\binBL}+\NBO{\binLSS} \geq 2 \left(\NBA{\binLLS,\binLL} - \NBO{\binLLS, \binLL} \right)
		$
	\end{lemma}
	
}
\begin{proof}
	Consider the number of $\itemL$ items in $\binLL$ and $\binLLS$ bins in both solutions.
	For each bin of type $\binLL$ or $\binLLS$ that is part of $\ALG$, but not of $\OPT$,
	two $\itemL$ items have to be in other bin types in the optimal solution.

	These items can be either in one $\binLSS$ bin each or they are paired with a $\itemB$ item to form a $\binBL$ bin.
	However, $\ALG$ is $\binBL$-thorough and hence has a maximum number of $\binBL$ bins given the pairings of $\itemB$ and $\itemL$ items in $\ALG$ are fixed.
	The additional $L$ items hence must be paired in $\OPT$ with items already paired in $\ALG$ which implies the existence of a $\binBL$ bin in $\ALG$ for each of these items.
\end{proof}

{
	\renewcommand{\thetheorem}{\ref{claim:boundS}}
	
	\begin{lemma}
		$
		\NBA{\binBS} + \NBA{\binLL} + 2\NBO{\binLSS}\geq 3 \left(\NBA{\binSSS} - \NBO{\binSSS} \right)
		$
	\end{lemma}
	
}
\begin{proof}
	We consider the number of $\itemS$ items in $\binSSS$ bins in both solutions.
	Again for each bin of type $\binSSS$ that is part of $\ALG$, but not of $\OPT$,
	three $\itemS$ items have to be in other bin types in the optimal solution.

	For sorting the items into $\itemB$ bins we use a similar argument to the proof of Lemma~\ref{claim:boundL}.
	Since $\ALG$ is $\binBS$-thorough the additional $\itemS$ items can only be together with $\itemB$ items which are combined with $\itemS$ items in $\ALG$.
	Recall that Lemma~\ref{lemma:optstructure} implies that the $\itemS$ items do not go with $\itemB$ items which were in $\binBL$ bins in $\ALG$.

	The $\itemS$ items can also be part of bins of type $\binLLS$.
	The number of $\itemL$ items available are upper bounded by the number of $\itemL$ items in
	$\binLL$ bins in $\ALG$ due to Lemma~\ref{lemma:optstructure}.

	Finally, the items can be part of $\binLSS$ bins in the optimal solution, where each $\binLSS$ bin can contain at most 2 of these items.
\end{proof}

{
	\renewcommand{\thetheorem}{\ref{claim:boundS2}}
	
	\begin{lemma}
		$
		\NBA{\binBS} + \NBO{\binLLS} + 2\NBO{\binLSS}\geq 3 \left(\NBA{\binSSS} - \NBO{\binSSS} \right)
		$
	\end{lemma}

}
\begin{proof}
	We take another look at the $\itemS$ items similar to Lemma~\ref{claim:boundS}.
	The combination with $\itemB$ items and the existence in $\binLSS$ bins is counted as before.
	For the number of additional $\binLLS$ bins in $\OPT$, we directly bound their amount by $\NBO{\binLLS}$.
\end{proof}

{
	\renewcommand{\thetheorem}{\ref{lemma:large_shifting_moves}}
	
	\begin{lemma}
		The number of shifting moves (regarding major items) is bounded by $\mathcal{O}(\nicefrac{1}{\varepsilon})$.
	\end{lemma}
	
}
\begin{proof}
	Ivkovic~\cite{Ivkovic1996} showed that the number of bins which need to be changed during an insertion or deletion procedure is a constant.
	Since we only treat items with size at least $\nicefrac{\varepsilon}{15}$ with this algorithm, the number of shifting moves
	can be upper bounded by $\mathcal{O}\left(\nicefrac{1}{\varepsilon}\right)$.
\end{proof}

\section{Proofs from Section~\ref{sec:Combination}}

{
	\renewcommand{\thetheorem}{\ref{lemma:shifting_moves}}
	
	\begin{lemma}
		The algorithm uses at most $\mathcal{O}(\nicefrac{1}{\varepsilon^2})$ shifting moves for each insertion or deletion.
	\end{lemma}
	
}
\begin{proof}
	If a major item is inserted or deleted, the MP algorithm only changes a constant number of bins in the solution for major items (cf. Lemma~\ref{lemma:large_shifting_moves} and~\cite{Ivkovic1996}).
	Hence, for each type of bin of minor items, only a constant number of bins must be recombined.
	In total, $\mathcal{O}(k)$ bins of major items are reallocated.
	Since each bin with major items contains at most $\nicefrac{15}{\varepsilon}$ items,
	this implies the bound on the shifting moves.
	
	If a minor item is inserted or deleted, according to Lemma~\ref{lemma:small_shifting_moves} the number of shifting moves to update the solution for minor items is $\mathcal{O}(\nicefrac{1}{\varepsilon^2})$.
	This process adds or deletes at most one bin if we only count the bins which contain at least one minor item.
	As a consequence, for each type of bin with minor items, at most one bin with major items must be reallocated.
	Hence we obtain the same bound of $\mathcal{O}(\nicefrac{1}{\varepsilon^2})$ for the number of shifting moves as before.
\end{proof}

{
	\renewcommand{\thetheorem}{\ref{lemma:combstructure}}
	
	\begin{lemma}
		Let $B^{\text{min}}$ be a bin with minor items of type $j$.
		If $B^{\text{min}}$ is not combined with a bin of major items, all non-combined bins with
		major items have a filling height of at least $1-w_j$.
	\end{lemma}
	
}
\begin{proof}
	If there is a bin $B^{\text{maj}}$ with major items which has not been combined then the algorithm
	would have attempted to combine it with $B^{\text{min}}$.
	Hence the filling height of $B^{\text{maj}}$ must be at least $1-w_j$.
\end{proof}

{
	\renewcommand{\thetheorem}{\ref{lemma:OPT_LL}}
	
	\begin{lemma}
		Let $\nonquartetLL\leq 2$. Then it holds that
		
		$\OPT\geq\frac{1}{\hat{z}_j}\cdot\left(\sum\limits_{\binVar\in\binSet}(1-z_j'\cdot F(\binVar))\cdot \NBA{\binVar}^{\geq(1-w_j);-Q} + (1-\frac{3}{4}z_j')\quartetLLSSS\right).$
	\end{lemma}
	
}
\begin{proof}
	To show this lemma, we start with the solution from $\OPT$ and estimate it with the help of the lemmas from Section~\ref{sec:JohnsonClaims}.
	Therefore, we split $\OPT$ into its components and apply the Lemmas one after another.
	In the following transformations we alternate between the use of the lemmas and simple mathematical calculations.
	At the end we get an estimation of the optimal solution which consists of different components of the algorithmic solution as well as specific components of the optimal solution.
	For the remaining parts of the optimal solution we show, that they are greater or equal to 0, for the components of $\ALG$ we show separately for each bin type, that they fulfill the statement of the lemma, i.e. we show they are at least $\frac{1}{\hat{z}_j}\cdot(1-z_j'\cdot F(\binVar))\cdot \NBA{\binVar}$.
	
	W.l.o.g. we assume $\nonquartetLL=0$ and hence showing a bound of $\frac{1-\frac{3}{4}z_{j}'}{\hat{z}_j}$ for the $\binLL$ bins is sufficient.
	Furthermore we show a bound of 1 for $\NBA{\binB}$ since $1=\frac{1-(1-y_j)z_j'}{\hat{z}_j}$.
	{\thinmuskip=0mu\medmuskip=0mu\thickmuskip=0mu\small\allowdisplaybreaks
	\begin{eqnarray*}
		&&\OPT \\
		&\geq& \NBO{\binB,\binBL,\binBS} + \NBO{\binLLS,\binLL,\binSSS}+\NBO{\binLSS} \\
		&\overset{\text{Lem.~\ref{claim:boundB}}}{\geq}& \NBA{\binB,\binBL,\binBS} + \NBO{\binLLS,\binLL,\binSSS}+\NBO{\binLSS} \\
		&=& \NBA{\binB} + \max\left\{\frac{1}{2},\frac{1-\frac{5}{6}z_j'}{\hat{z}_j}\right\}\cdot\NBA{\binBL} + \NBA{\binBS}+ \NBO{\binLLS,\binLL,\binSSS} \\
		&& +\max\left\{\frac{1}{2},\frac{1-\frac{5}{6}z_j'}{\hat{z}_j}\right\}\cdot\NBO{\binLSS}
		+ \left(1-\max\left\{\frac{1}{2},\frac{1-\frac{5}{6}z_j'}{\hat{z}_j}\right\}\right)\cdot\left(\NBA{\binBL}+\NBO{\binLSS}\right) \\
		&\overset{\text{Lem.~\ref{claim:boundL}}}{\geq}& \NBA{\binB} + \max\left\{\frac{1}{2},\frac{1-\frac{5}{6}z_j'}{\hat{z}_j}\right\}\cdot\NBA{\binBL} + \NBA{\binBS}+ \NBO{\binLLS,\binLL,\binSSS} \\
		&& +\max\left\{\frac{1}{2},\frac{1-\frac{5}{6}z_j'}{\hat{z}_j}\right\}\cdot\NBO{\binLSS} + \left(1-\max\left\{\frac{1}{2},\frac{1-\frac{5}{6}z_j'}{\hat{z}_j}\right\}\right)\cdot 2\cdot\left(\NBA{\binLLS,\binLL}-\NBO{\binLLS,\binLL}\right) \\
		&=& \NBA{\binB} + \max\left\{\frac{1}{2},\frac{1-\frac{5}{6}z_j'}{\hat{z}_j}\right\}\cdot\NBA{\binBL} 
		+ \left(2\max\left\{\frac{1}{2},\frac{1-\frac{5}{6}z_j'}{\hat{z}_j}\right\} + \max\left\{\frac{3}{4},\frac{1-\frac{3}{4}z_j'}{\hat{z}_j}\right\}-1\right)\cdot\NBA{\binBS} \\
		&& + \left(2-2\max\left\{\frac{1}{2},\frac{1-\frac{5}{6}z_j'}{\hat{z}_j}\right\}\right)\cdot\NBA{\binLLS} + \max\left\{\frac{3}{4},\frac{1-\frac{3}{4}z_j'}{\hat{z}_j}\right\}\cdot\NBA{\binLL} \\
		&& + \left(2\max\left\{\frac{1}{2},\frac{1-\frac{5}{6}z_j'}{\hat{z}_j}\right\}-1\right)\cdot\NBO{\binLLS,\binLL} +\NBO{\binSSS} \\
		&& + \left(5\max\left\{\frac{1}{2},\frac{1-\frac{5}{6}z_j'}{\hat{z}_j}\right\}-4+2\max\left\{\frac{3}{4},\frac{1-\frac{3}{4}z_j'}{\hat{z}_j}\right\}\right)\cdot\NBO{\binLSS} \\
		&& + \left(2-2\max\left\{\frac{1}{2},\frac{1-\frac{5}{6}z_j'}{\hat{z}_j}\right\} - \max\left\{\frac{3}{4},\frac{1-\frac{3}{4}z_j'}{\hat{z}_j}\right\}\right)\cdot\left(\NBA{\binBS} + \NBA{\binLL}+2\NBO{\binLSS}\right)\\
		&\overset{\text{Lem.~\ref{claim:boundS}}}{\geq}& \NBA{\binB} + \max\left\{\frac{1}{2},\frac{1-\frac{5}{6}z_j'}{\hat{z}_j}\right\}\cdot\NBA{\binBL} 
		+ \left(2\max\left\{\frac{1}{2},\frac{1-\frac{5}{6}z_j'}{\hat{z}_j}\right\} + \max\left\{\frac{3}{4},\frac{1-\frac{3}{4}z_j'}{\hat{z}_j}\right\}-1\right)\cdot\NBA{\binBS} \\
		&& + \left(2-2\max\left\{\frac{1}{2},\frac{1-\frac{5}{6}z_j'}{\hat{z}_j}\right\}\right)\cdot\NBA{\binLLS} + \max\left\{\frac{3}{4},\frac{1-\frac{3}{4}z_j'}{\hat{z}_j}\right\}\cdot\NBA{\binLL} \\
		&& + \left(2\max\left\{\frac{1}{2},\frac{1-\frac{5}{6}z_j'}{\hat{z}_j}\right\}-1\right)\cdot\NBO{\binLLS,\binLL} + \NBO{\binSSS} \\
		&& + \left(5\max\left\{\frac{1}{2},\frac{1-\frac{5}{6}z_j'}{\hat{z}_j}\right\}-4+2\max\left\{\frac{3}{4},\frac{1-\frac{3}{4}z_j'}{\hat{z}_j}\right\}\right)\cdot\NBO{\binLSS}\\
		&& + \left(2-2\max\left\{\frac{1}{2},\frac{1-\frac{5}{6}z_j'}{\hat{z}_j}\right\} - \max\left\{\frac{3}{4},\frac{1-\frac{3}{4}z_j'}{\hat{z}_j}\right\}\right)\cdot3\cdot\left(\NBA{\binSSS} - \NBO{\binSSS}\right)\\
		&=& \NBA{\binB} + \max\left\{\frac{1}{2},\frac{1-\frac{5}{6}z_j'}{\hat{z}_j}\right\}\cdot\NBA{\binBL}
		+ \max\left\{\frac{3}{4},\frac{1-\frac{3}{4}z_j'}{\hat{z}_j}\right\}\cdot\NBA{\binBS} \\
		&&	+ \left(2-2\max\left\{\frac{1}{2},\frac{1-\frac{5}{6}z_j'}{\hat{z}_j}\right\}\right)\cdot\NBA{\binLLS} + \max\left\{\frac{3}{4},\frac{1-\frac{3}{4}z_j'}{\hat{z}_j}\right\}\cdot\NBA{\binLL} \\
		&&  + \left(6-6\max\left\{\frac{1}{2},\frac{1-\frac{5}{6}z_j'}{\hat{z}_j}\right\} - 3\max\left\{\frac{3}{4},\frac{1-\frac{3}{4}z_j'}{\hat{z}_j}\right\}\right)\cdot\NBA{\binSSS}\\
		&& + \left(2\max\left\{\frac{1}{2},\frac{1-\frac{5}{6}z_j'}{\hat{z}_j}\right\}-1\right)\cdot\NBO{\binLL}\\
		&&  + \left(6\max\left\{\frac{1}{2},\frac{1-\frac{5}{6}z_j'}{\hat{z}_j}\right\} + 3\max\left\{\frac{3}{4},\frac{1-\frac{3}{4}z_j'}{\hat{z}_j}\right\}-5\right)\cdot\NBO{\binSSS} \\
		&&  + \left(\max\left\{\frac{1}{2},\frac{1-\frac{5}{6}z_j'}{\hat{z}_j}\right\}-2+2\max\left\{\frac{3}{4},\frac{1-\frac{3}{4}z_j'}{\hat{z}_j}\right\}\right)\cdot\NBO{\binLSS} \\
		&&  + \left(2\max\left\{\frac{1}{2},\frac{1-\frac{5}{6}z_j'}{\hat{z}_j}\right\}-1\right)\cdot\left(\NBA{\binBS}+\NBO{\binLLS}+2\NBO{LSS}\right) \\
		&\overset{\text{Lem.~\ref{claim:boundS2}}}{\geq}& \NBA{\binB} + \max\left\{\frac{1}{2},\frac{1-\frac{5}{6}z_j'}{\hat{z}_j}\right\}\cdot\NBA{\binBL}
		+ \max\left\{\frac{3}{4},\frac{1-\frac{3}{4}z_j'}{\hat{z}_j}\right\}\cdot\NBA{\binBS} \\
		&&	+ \left(2-2\max\left\{\frac{1}{2},\frac{1-\frac{5}{6}z_j'}{\hat{z}_j}\right\}\right)\cdot\NBA{\binLLS} + \max\left\{\frac{3}{4},\frac{1-\frac{3}{4}z_j'}{\hat{z}_j}\right\}\cdot\NBA{\binLL} \\
		&&  + \left(6-6\max\left\{\frac{1}{2},\frac{1-\frac{5}{6}z_j'}{\hat{z}_j}\right\} - 3\max\left\{\frac{3}{4},\frac{1-\frac{3}{4}z_j'}{\hat{z}_j}\right\}\right)\cdot\NBA{\binSSS}\\
		&& + \left(2\max\left\{\frac{1}{2},\frac{1-\frac{5}{6}z_j'}{\hat{z}_j}\right\}-1\right)\cdot\NBO{\binLL}\\
		&&  + \left(6\max\left\{\frac{1}{2},\frac{1-\frac{5}{6}z_j'}{\hat{z}_j}\right\} + 3\max\left\{\frac{3}{4},\frac{1-\frac{3}{4}z_j'}{\hat{z}_j}\right\}-5\right)\cdot\NBO{\binSSS} \\
		&&  +  \left(\max\left\{\frac{1}{2},\frac{1-\frac{5}{6}z_j'}{\hat{z}_j}\right\}-2+2\max\left\{\frac{3}{4},\frac{1-\frac{3}{4}z_j'}{\hat{z}_j}\right\}\right)\cdot\NBO{\binLSS}\\
		&&  + \left(2\max\left\{\frac{1}{2},\frac{1-\frac{5}{6}z_j'}{\hat{z}_j}\right\}-1\right)\cdot3\cdot\left(\NBA{\binSSS} - \NBO{\binSSS}\right) \\
		&=& \NBA{\binB} + \max\left\{\frac{1}{2},\frac{1-\frac{5}{6}z_j'}{\hat{z}_j}\right\}\cdot\NBA{\binBL}
		+ \max\left\{\frac{3}{4},\frac{1-\frac{3}{4}z_j'}{\hat{z}_j}\right\}\cdot\NBA{\binBS} \\
		&&	+ \left(2-2\max\left\{\frac{1}{2},\frac{1-\frac{5}{6}z_j'}{\hat{z}_j}\right\}\right)\cdot\NBA{\binLLS} + \max\left\{\frac{3}{4},\frac{1-\frac{3}{4}z_j'}{\hat{z}_j}\right\}\cdot\NBA{\binLL} \\
		&&  + \left(3 - 3\max\left\{\frac{3}{4},\frac{1-\frac{3}{4}z_j'}{\hat{z}_j}\right\}\right)\cdot\NBA{\binSSS} + \left(2\max\left\{\frac{1}{2},\frac{1-\frac{5}{6}z_j'}{\hat{z}_j}\right\}-1\right)\cdot\NBO{\binLL}\\
		&& + \left(3\max\left\{\frac{3}{4},\frac{1-\frac{3}{4}z_j'}{\hat{z}_j}\right\}-2\right)\cdot\NBO{\binSSS} \\
		&&  +  \left(\max\left\{\frac{1}{2},\frac{1-\frac{5}{6}z_j'}{\hat{z}_j}\right\}-2+2\max\left\{\frac{3}{4},\frac{1-\frac{3}{4}z_j'}{\hat{z}_j}\right\}\right)\cdot\NBO{\binLSS}
	\end{eqnarray*}
	}

	As explained above we now have to guarantee that the remaining components of the optimal solution are greater or equal than 0.
	This is obviously given, since it holds that $2\max\left\{\frac{1}{2},\frac{1-\frac{5}{6}z_j'}{\hat{z}_j}\right\}-1\geq 0$, $3\cdot\max\left\{\frac{3}{4},\frac{1-\frac{3}{4}z_j'}{\hat{z}_j}\right\}-2\geq 0$ and $\max\left\{\frac{1}{2},\frac{1-\frac{5}{6}z_j'}{\hat{z}_j}\right\}-2+2\max\left\{\frac{3}{4},\frac{1-\frac{3}{4}z_j'}{\hat{z}_j}\right\}\geq 0$.
	
	It remains to show that for the components of the algorithmic solution, each bin type fulfills its corresponding limit.
	For the bin types $\binB, \binBL, \binBS$ and $\binLL$ this is directly visible because of their corresponding $F(\binVar)$, so it explicitly has to be shown for the types $\binLLS$ and $\binSSS$:
	\begin{enumerate}
		\item $2-2\max\left\{\frac{1}{2},\frac{1-\frac{5}{6}z_j'}{\hat{z}_j}\right\}\geq \frac{1-\frac{11}{12}z_{j}'}{\hat{z}_j}$:
		In case the maximum on the left hand side is $\frac{1}{2}$, this is obvious.
		Otherwise
		$$2-2\frac{1-\frac{5}{6}z_j'}{\hat{z}_j}\geq \frac{1-\frac{11}{12}z_{j}'}{\hat{z}_j} \Leftrightarrow 2z_k'+\frac{7}{12}z_j'\geq 3$$
		which holds since $z_k'\geq\frac{4}{3}$ and $z_j'\geq\frac{3}{4}$.
		\item $3 - 3\max\left\{\frac{3}{4},\frac{1-\frac{3}{4}z_j'}{\hat{z}_j}\right\}\geq \frac{1-\frac{3}{4}z_{j}'}{\hat{z}_j}$:
		The claim follows directly from $\frac{1-\frac{3}{4}z_{j}'}{\hat{z}_j}\leq \frac{3}{4}$.
	\end{enumerate}
	
\end{proof}

{
	\renewcommand{\thetheorem}{\ref{lemma:OPT_SSS}}
	
	\begin{lemma}
		Let $\nonquartetSSS=0$. Then it holds that
		
		$\OPT\geq\frac{1}{\hat{z}_j}\cdot\left(\sum\limits_{\binVar\in\binSet}(1-z_j'\cdot F(\binVar))\cdot \NBA{\binVar}^{\geq(1-w_j);-Q} + (4-3z_j')\quartet\right).$
	\end{lemma}
	
}
\begin{proof}
	Again, we use the same idea as in the proof of Lemma~\ref{lemma:OPT_LL} and estimate $\OPT$ with the lemmas from Section~\ref{sec:JohnsonClaims}:
	
	{\thinmuskip=0mu\medmuskip=0mu\thickmuskip=0mu\allowdisplaybreaks\small
	\begin{eqnarray*}
		&&\OPT\\
		&\geq& \NBO{\binB,\binBL,\binBS} + \NBO{\binLLS,\binLL,\binSSS} + \NBO{\binLSS} \\
		&\overset{\text{Lem.~\ref{claim:boundB}}}{\geq}& \NBA{\binB,\binBL,\binBS} + \NBO{\binLLS,\binLL,\binSSS} + \NBO{\binLSS} \\
		&=& \NBA{\binB} + \max\left\{\frac{1}{2},\frac{1-\frac{5}{6}z_j'}{\hat{z}_j}\right\}\cdot\NBA{\binBL} + \NBA{\binBS}+ \NBO{\binLLS,\binLL,\binSSS} \\
		&&+\max\left\{\frac{1}{2},\frac{1-\frac{5}{6}z_j'}{\hat{z}_j}\right\}\cdot\NBO{\binLSS}
		+ \left(1-\max\left\{\frac{1}{2},\frac{1-\frac{5}{6}z_j'}{\hat{z}_j}\right\}\right)\cdot\left(\NBA{\binBL}+\NBO{\binLSS}\right) \\
		&\overset{\text{Lem.~\ref{claim:boundL}}}{\geq}& \NBA{\binB} + \max\left\{\frac{1}{2},\frac{1-\frac{5}{6}z_j'}{\hat{z}_j}\right\}\cdot\NBA{\binBL} + \NBA{\binBS}+ \NBO{\binLLS,\binLL,\binSSS} \\
		&& +  \max\left\{\frac{1}{2},\frac{1-\frac{5}{6}z_j'}{\hat{z}_j}\right\}\cdot\NBO{\binLSS}
		+ \left(1-\max\left\{\frac{1}{2},\frac{1-\frac{5}{6}z_j'}{\hat{z}_j}\right\}\right)\cdot 2\cdot\left(\NBA{\binLLS,\binLL}-\NBO{\binLLS,\binLL}\right)\\
		&=& \NBA{\binB} + \max\left\{\frac{1}{2},\frac{1-\frac{5}{6}z_j'}{\hat{z}_j}\right\}\cdot\NBA{\binBL} \\
		&&+ \left(2\max\left\{\frac{1}{2},\frac{1-\frac{5}{6}z_j'}{\hat{z}_j}\right\} + \max\left\{\frac{3}{4},\frac{1-\max\left\{\frac{2}{3},1-w_j\right\}z_j'}{\hat{z}_j}\right\}-1\right)\cdot\NBA{\binBS} \\
		&& + \left(2-2\max\left\{\frac{1}{2},\frac{1-\frac{5}{6}z_j'}{\hat{z}_j}\right\}\right)\cdot\NBA{\binLLS} + \max\left\{\frac{3}{4},\frac{1-\max\left\{\frac{2}{3},1-w_j\right\}z_j'}{\hat{z}_j}\right\}\cdot\NBA{\binLL} \\
		&& + \left(2\max\left\{\frac{1}{2},\frac{1-\frac{5}{6}z_j'}{\hat{z}_j}\right\}-1\right)\cdot\NBO{\binLLS,\binLL} +\NBO{\binSSS}  \\
		&& + \left(5\max\left\{\frac{1}{2},\frac{1-\frac{5}{6}z_j'}{\hat{z}_j}\right\}-4+2\max\left\{\frac{3}{4},\frac{1-\max\left\{\frac{2}{3},1-w_j\right\}z_j'}{\hat{z}_j}\right\}\right)\cdot\NBO{\binLSS}\\
		&& + \left(2-2\max\left\{\frac{1}{2},\frac{1-\frac{5}{6}z_j'}{\hat{z}_j}\right\} - \max\left\{\frac{3}{4},\frac{1-\max\left\{\frac{2}{3},1-w_j\right\}z_j'}{\hat{z}_j}\right\}\right)\cdot\left(\NBA{\binBS} + \NBA{\binLL}+ 2\NBO{\binLSS}\right) \\
		&\overset{\text{Lem.~\ref{claim:boundS}}}{\geq}& \NBA{\binB} + \max\left\{\frac{1}{2},\frac{1-\frac{5}{6}z_j'}{\hat{z}_j}\right\}\cdot\NBA{\binBL} \\
		&& + \left(2\max\left\{\frac{1}{2},\frac{1-\frac{5}{6}z_j'}{\hat{z}_j}\right\} + \max\left\{\frac{3}{4},\frac{1-\max\left\{\frac{2}{3},1-w_j\right\}z_j'}{\hat{z}_j}\right\}-1\right)\cdot\NBA{\binBS} \\
		&& + \left(2-2\max\left\{\frac{1}{2},\frac{1-\frac{5}{6}z_j'}{\hat{z}_j}\right\}\right)\cdot\NBA{\binLLS} + \max\left\{\frac{3}{4},\frac{1-\max\left\{\frac{2}{3},1-w_j\right\}z_j'}{\hat{z}_j}\right\}\cdot\NBA{\binLL} \\
		&& + \left(2\max\left\{\frac{1}{2},\frac{1-\frac{5}{6}z_j'}{\hat{z}_j}\right\}-1\right)\cdot\NBO{\binLLS,\binLL} + \NBO{\binSSS} \\
		&& +\left(5\max\left\{\frac{1}{2},\frac{1-\frac{5}{6}z_j'}{\hat{z}_j}\right\}-4+2\max\left\{\frac{3}{4},\frac{1-\max\left\{\frac{2}{3},1-w_j\right\}z_j'}{\hat{z}_j}\right\}\right)\cdot\NBO{\binLSS} \\
		&& + \left(2-2\max\left\{\frac{1}{2},\frac{1-\frac{5}{6}z_j'}{\hat{z}_j}\right\} - \max\left\{\frac{3}{4},\frac{1-\max\left\{\frac{2}{3},1-w_j\right\}z_j'}{\hat{z}_j}\right\}\right)\cdot3\cdot\left(\NBA{\binSSS} - \NBO{\binSSS}\right)\\
		&=& \NBA{\binB} + \max\left\{\frac{1}{2},\frac{1-\frac{5}{6}z_j'}{\hat{z}_j}\right\}\cdot\NBA{\binBL}
		+ \max\left\{\frac{3}{4},\frac{1-\max\left\{\frac{2}{3},1-w_j\right\}z_j'}{\hat{z}_j}\right\}\cdot\NBA{\binBS} \\
		&&	+ \left(2-2\max\left\{\frac{1}{2},\frac{1-\frac{5}{6}z_j'}{\hat{z}_j}\right\}\right)\cdot\NBA{\binLLS} + \max\left\{\frac{3}{4},\frac{1-\max\left\{\frac{2}{3},1-w_j\right\}z_j'}{\hat{z}_j}\right\}\cdot\NBA{\binLL} \\
		&&  + \left(6-6\max\left\{\frac{1}{2},\frac{1-\frac{5}{6}z_j'}{\hat{z}_j}\right\} - 3\max\left\{\frac{3}{4},\frac{1-\max\left\{\frac{2}{3},1-w_j\right\}z_j'}{\hat{z}_j}\right\}\right)\cdot\NBA{\binSSS} \\
		&&+ \left(2\max\left\{\frac{1}{2},\frac{1-\frac{5}{6}z_j'}{\hat{z}_j}\right\}-1\right)\cdot\NBO{\binLL} \\
		&&  + \left(6\max\left\{\frac{1}{2},\frac{1-\frac{5}{6}z_j'}{\hat{z}_j}\right\} + 3\max\left\{\frac{3}{4},\frac{1-\max\left\{\frac{2}{3},1-w_j\right\}z_j'}{\hat{z}_j}\right\}-5\right)\cdot\NBO{\binSSS} \\
		&&  +  \left(\max\left\{\frac{1}{2},\frac{1-\frac{5}{6}z_j'}{\hat{z}_j}\right\}-2+2\max\left\{\frac{3}{4},\frac{1-\max\left\{\frac{2}{3},1-w_j\right\}z_j'}{\hat{z}_j}\right\}\right)\cdot\NBO{\binLSS}\\
		&&  + \left(2\max\left\{\frac{1}{2},\frac{1-\frac{5}{6}z_j'}{\hat{z}_j}\right\}-1\right)\cdot\left(\NBA{\binBS}+\NBO{\binLLS}+2\NBO{LSS}\right) \\
		&\overset{\text{Lem.~\ref{claim:boundS2}}}{\geq}& \NBA{\binB} + \max\left\{\frac{1}{2},\frac{1-\frac{5}{6}z_j'}{\hat{z}_j}\right\}\cdot\NBA{\binBL}
		+ \max\left\{\frac{3}{4},\frac{1-\max\left\{\frac{2}{3},1-w_j\right\}z_j'}{\hat{z}_j}\right\}\cdot\NBA{\binBS} \\
		&&	+ \left(2-2\max\left\{\frac{1}{2},\frac{1-\frac{5}{6}z_j'}{\hat{z}_j}\right\}\right)\cdot\NBA{\binLLS} + \max\left\{\frac{3}{4},\frac{1-\max\left\{\frac{2}{3},1-w_j\right\}z_j'}{\hat{z}_j}\right\}\cdot\NBA{\binLL} \\
		&&  + \left(6-6\max\left\{\frac{1}{2},\frac{1-\frac{5}{6}z_j'}{\hat{z}_j}\right\} - 3\max\left\{\frac{3}{4},\frac{1-\max\left\{\frac{2}{3},1-w_j\right\}z_j'}{\hat{z}_j}\right\}\right)\cdot\NBA{\binSSS} \\
		&&  + \left(2\max\left\{\frac{1}{2},\frac{1-\frac{5}{6}z_j'}{\hat{z}_j}\right\}-1\right)\cdot\NBO{\binLL} \\
		&&  + \left(6\max\left\{\frac{1}{2},\frac{1-\frac{5}{6}z_j'}{\hat{z}_j}\right\} + 3\max\left\{\frac{3}{4},\frac{1-\max\left\{\frac{2}{3},1-w_j\right\}z_j'}{\hat{z}_j}\right\}-5\right)\cdot\NBO{\binSSS} \\
		&&  + \left(\max\left\{\frac{1}{2},\frac{1-\frac{5}{6}z_j'}{\hat{z}_j}\right\}-2+2\max\left\{\frac{3}{4},\frac{1-\max\left\{\frac{2}{3},1-w_j\right\}z_j'}{\hat{z}_j}\right\}\right)\cdot\NBO{\binLSS}\\
		&&  + \left(2\max\left\{\frac{1}{2},\frac{1-\frac{5}{6}z_j'}{\hat{z}_j}\right\}-1\right)\cdot3\cdot\left(\NBA{\binSSS} - \NBO{\binSSS}\right) \\
		&=& \NBA{\binB} + \max\left\{\frac{1}{2},\frac{1-\frac{5}{6}z_j'}{\hat{z}_j}\right\}\cdot\NBA{\binBL}\\
		&&+ \max\left\{\frac{3}{4},\frac{1-\max\left\{\frac{2}{3},1-w_j\right\}z_j'}{\hat{z}_j}\right\}\cdot\NBA{\binBS} \\
		&&	+ \left(2-2\max\left\{\frac{1}{2},\frac{1-\frac{5}{6}z_j'}{\hat{z}_j}\right\}\right)\cdot\NBA{\binLLS} + \max\left\{\frac{3}{4},\frac{1-\max\left\{\frac{2}{3},1-w_j\right\}z_j'}{\hat{z}_j}\right\}\cdot\left(\quartetLL+\nonquartetLL\right) \\
		&&  + \left(3 - 3\max\left\{\frac{3}{4},\frac{1-\max\left\{\frac{2}{3},1-w_j\right\}z_j'}{\hat{z}_j}\right\}\right)\cdot\quartetSSS 
		  + \left(2\max\left\{\frac{1}{2},\frac{1-\frac{5}{6}z_j'}{\hat{z}_j}\right\}-1\right)\cdot\NBO{\binLL} \\
		&& + \left(3\max\left\{\frac{3}{4},\frac{1-\max\left\{\frac{2}{3},1-w_j\right\}z_j'}{\hat{z}_j}\right\}-2\right)\cdot\NBO{\binSSS} \\
		&&  + \left(\max\left\{\frac{1}{2},\frac{1-\frac{5}{6}z_j'}{\hat{z}_j}\right\}-2+2\max\left\{\frac{3}{4},\frac{1-\max\left\{\frac{2}{3},1-w_j\right\}z_j'}{\hat{z}_j}\right\}\right)\cdot\NBO{\binLSS} \\
		&=& \NBA{\binB} + \max\left\{\frac{1}{2},\frac{1-\frac{5}{6}z_j'}{\hat{z}_j}\right\}\cdot\NBA{\binBL}
		+ \max\left\{\frac{3}{4},\frac{1-\max\left\{\frac{2}{3},1-w_j\right\}z_j'}{\hat{z}_j}\right\}\cdot\NBA{\binBS} \\
		&&	+ \left(2-2\max\left\{\frac{1}{2},\frac{1-\frac{5}{6}z_j'}{\hat{z}_j}\right\}\right)\cdot\NBA{\binLLS} + \max\left\{\frac{3}{4},\frac{1-\max\left\{\frac{2}{3},1-w_j\right\}z_j'}{\hat{z}_j}\right\}\cdot\left(3\quartet+\nonquartetLL\right) \\
		&&  + \left(3 - 3\max\left\{\frac{3}{4},\frac{1-\max\left\{\frac{2}{3},1-w_j\right\}z_j'}{\hat{z}_j}\right\}\right)\cdot\quartet 
		  + \left(2\max\left\{\frac{1}{2},\frac{1-\frac{5}{6}z_j'}{\hat{z}_j}\right\}-1\right)\cdot\NBO{\binLL} \\
		&& + \left(3\max\left\{\frac{3}{4},\frac{1-\max\left\{\frac{2}{3},1-w_j\right\}z_j'}{\hat{z}_j}\right\}-2\right)\cdot\NBO{\binSSS} \\
		&&  + \left(\max\left\{\frac{1}{2},\frac{1-\frac{5}{6}z_j'}{\hat{z}_j}\right\}-2+2\max\left\{\frac{3}{4},\frac{1-\max\left\{\frac{2}{3},1-w_j\right\}z_j'}{\hat{z}_j}\right\}\right)\cdot\NBO{\binLSS} \\
		&\geq& \NBA{\binB} + \max\left\{\frac{1}{2},\frac{1-\frac{5}{6}z_j'}{\hat{z}_j}\right\}\cdot\NBA{\binBL}
		+ \max\left\{\frac{3}{4},\frac{1-\max\left\{\frac{2}{3},1-w_j\right\}z_j'}{\hat{z}_j}\right\}\cdot\NBA{\binBS} \\
		&&	+ \left(2-2\max\left\{\frac{1}{2},\frac{1-\frac{5}{6}z_j'}{\hat{z}_j}\right\}\right)\cdot\NBA{\binLLS} + \max\left\{\frac{3}{4},\frac{1-\max\left\{\frac{2}{3},1-w_j\right\}z_j'}{\hat{z}_j}\right\}\cdot\nonquartetLL + 3\cdot\quartet
	\end{eqnarray*}
	}
	Note that we used $\quartet=\quartetSSS=\nicefrac{1}{3}\cdot \quartetLL$.
	
	The remaining components of the optimal solution are obviously greater or equal than $0$ and omitted in the last inequality. For the bin types $\binB$ and $\binBL$ the required limits are directly given. To complete the proof we consider the bin types $\binBS, \binLLS, \binLL$ and the quartets and we have to show the following to complete the proof:
	\begin{enumerate}
		\item $\max\left\{\frac{3}{4},\frac{1-\max\left\{\frac{2}{3},1-w_j\right\}z_j'}{\hat{z}_j}\right\}\geq\frac{1-\frac{3}{4}z_j'}{\hat{z}_j}$:
		It suffices to show $1-(1-w_j)z_j'\geq 1-\frac{3}{4}z_j'$ which holds since $w_j\geq\frac{1}{4}$ and $1-\frac{2}{3}z_j'\geq1-\frac{3}{4}z_j'$ which holds trivially.
		\item $2-2\max\left\{\frac{1}{2},\frac{1-\frac{5}{6}z_j'}{\hat{z}_j}\right\}\geq\frac{1-\frac{11}{12}z_j'}{\hat{z}_j}$ holds as in the previous lemma.
		\item $3\geq\frac{4-3z_j'}{\hat{z}_j}$ is equivalent to $3z_k'\geq 4$ which holds since $z_k'\geq\frac{4}{3}$.
	\end{enumerate}
	
\end{proof}

\end{document}